\documentclass[11pt]{article}
\usepackage{amsfonts}
\usepackage[latin9]{inputenc}
\usepackage[margin=1in]{geometry}
\usepackage{setspace}
\usepackage{amsmath}
\usepackage{amssymb}
\usepackage{bbm}
\usepackage{dcolumn}
\usepackage{graphicx}
\usepackage{enumerate}
\usepackage[inline, shortlabels]{enumitem}
\usepackage[multiple]{footmisc}
\usepackage{hhline}
\usepackage[position=bottom]{subfig}
\usepackage[authoryear]{natbib}
\usepackage[colorlinks,linkcolor=blue,citecolor=blue]{hyperref}
\usepackage{algorithm,algorithmic}
\usepackage{verbatim}
\usepackage{placeins}

\newcommand{\norm}[1]{\left\lVert#1\right\rVert}

\newcolumntype{d}{D{.}{.}{-1}}
\newtheorem{theorem}{Theorem}
\newtheorem{remark}{Remark}
\newtheorem{lemma}{Lemma}
\newtheorem{corollary}{Corollary}
\newtheorem{definition}{Definition}
\newtheorem{example}{Example}
\newtheorem{assumption}{Assumption}
\newtheorem{proposition}{Proposition}

\newenvironment{proof}[1][Proof]{\noindent \textbf{#1.} }{\  \rule{0.5em}{0.5em}}

\title{Local Asymptotic Power of Honest Confidence Intervals}
\author{Hugo Freeman\footnote{
                   Department of Economics,
                   Michigan State University,
                   East Lansing,
                   USA.
                   }
}
\begin{document}

\maketitle

\begin{abstract}
    Honest or bias-aware confidence intervals, conservative against untestable bias, are now standard in DiD, IV, RD, and factor models. This paper characterises the local power of their tests, which has three regimes set by the bias bound's rate relative to the sampling rate. A bound vanishing faster than the standard error costs nothing, one of the same order costs a bounded amount, and a dominating bound, typical at the parametric rate, leaves zero local power. This loss is intrinsic to honesty, and the standard bias-aware interval is rate-optimal. The paper recommends reporting the power ``dead zone'' half-width alongside bias-aware intervals.
\end{abstract}


\newpage

\section{Introduction}

Applied work increasingly reports confidence intervals widened by a worst-case bound on an untestable bias, so that coverage holds uniformly over the data-generating processes the data cannot rule out. These \emph{honest} or \emph{bias-aware} intervals are now standard for difference-in-differences under possible violations of parallel trends \citep{rambachan2023more}, instrumental variables under possible exclusion failures \citep{conley2012plausibly}, regression discontinuity \citep{armstrong2018optimal,armstrong2020simple}, and panels with potentially weak factors \citep{armstrong2022robust}.
This paper asks what such intervals cost in power and shows the cost can be total, since at the parametric rate an honest interval may cover every local alternative with probability approaching one. A minimax argument (Theorem~\ref{thm:lowerbound}) shows the loss is intrinsic to honesty, not a flaw of one procedure. Conversely, the standard bias-aware interval is rate-optimal among honest procedures (Corollary~\ref{cor:rateoptimal}), and the directional critical value of \citet{armstrong2018optimal} attains the sharp honest power frontier of Proposition~\ref{prop:sharp}, so the results price honesty rather than any particular procedure.

The idea that organises the results is that local power is set by the limit of the rescaled interval width, so the power loss is about width that does not shrink. When the rescaled interval converges to a set of positive volume, the \emph{dual} test, which rejects a value the interval excludes, has no power inside that set and at best one-sided power at its boundary. Partial identification, with its fixed identified set, is the limiting case \citep[cf.][]{imbens2004confidence,stoye2009more}, and honest inference is the same phenomenon at a slower rate, the bias bound playing the role of the identified-set width.

\citet{armstrong2018optimal,armstrong2020simple} and \citet{imbens2019optimized} 
construct bias-aware intervals from estimator $\hat\beta$, and bias bound $\sup_\theta |\hat{B}|$ such that
\begin{align*}
    \hat\beta - \beta = \hat{B} + O_p(n^{-\epsilon}),
\end{align*}
where $\hat B$ is the true bias, the bound $\sup_\theta |\hat{B}|$ depends on untestable user-defined parameters $\theta$, and $O_p(n^{-\epsilon})$ is sampling error, with $\epsilon = 1/2$ in a parametric model. Knowing only the bound, the analyst cannot debias, and the bound may be conservative, since the true bias can be negligible, $\hat B = o_p(n^{-\epsilon})$, while the bound may not be.
Assumed throughout is that $\sup_\theta |\hat{B}|$ is at least conservative, i.e. $|\hat B|\leq \sup_\theta |\hat{B}|$. 

The central object is the dead zone, the true effects near the hypothesised value that an honest test rejects no more often than its size. Ignoring sampling error, $\hat\beta$ is centred at the truth plus a bias of at most $\sup_\theta|\hat B|$, and an honest interval built from $\hat\beta$ must keep every value within that bound of the centre since each of them could be the truth under some admissible bias. The bias is untestable and may take either sign, so the truth can sit anywhere in this band, and rejection cannot be guaranteed for values up to twice the bound either side of it, exactly so in the leading applied settings (Figure~\ref{fig:deadzone}). The dead zone has the same width about any hypothesised value, and reporting its half-width, twice the bias bound, is the paper's headline disclosure recommendation.

Take the following toy example. Assume $\sup_\theta|\hat B|=1$, realised $\hat B = 1$, and no sampling error. The honest interval $\hat\beta\pm1=[\beta_0,\beta_0+2]$ covers the truth $\beta_0$ at its edge, and the width is well spent. However, when $\hat B=0$, the bound is conservative and the power cost is genuine. The interval $[\beta_0-1,\beta_0+1]$ spends $2\sup_\theta|\hat B|$, whereas an oracle knowing $\hat B=0$ would spend only $2z_{1-\alpha/2}\,se(\hat\beta)$. Because the bias is untestable, the saving can never be harvested, so the analyst pays for the full bound in both directions whatever the true bias, and the dead zone is that payment in power.

Asymptotically, the dead zone closes if the bound shrinks faster than the local alternatives, and conservatism is free. At the matched rate it persists at a fixed, explicit width and the power loss is bounded (Theorem~\ref{thm:nondeg}, Proposition~\ref{prop:sharp}). If the bound shrinks more slowly, as is typical in the parametric case, local power is \emph{zero} (Theorems~\ref{thm:LocalPower}--\ref{thm:LocalCI2}). The factor-model, difference-in-differences, and instrumental-variables examples fall into the zero power category. Power loss for honest regression discontinuity, the bounded loss case, is developed in Sections~\ref{sect:minimaxpowerLossHonest} and~\ref{sect:primitives}.

\begin{figure}[!ht]
\centering
\includegraphics[width=0.96\textwidth]{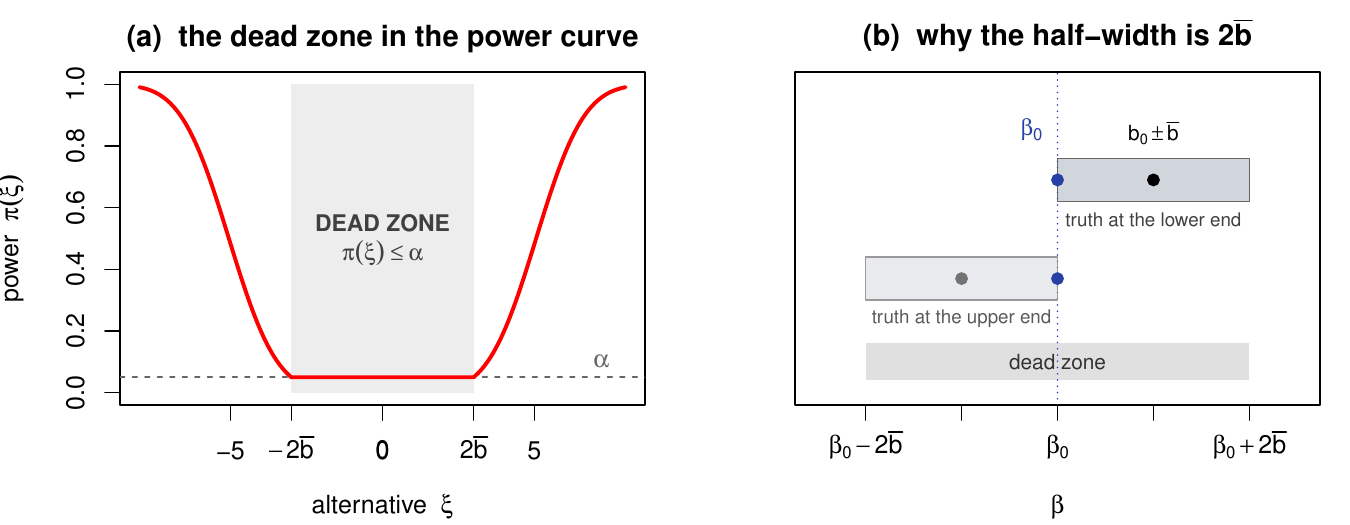}
\caption{The power dead zone. \textbf{(a)} The sharp honest power frontier
$\pi_\alpha(\xi)=\Phi\big((|\xi|-2\bar b)_+/\sigma-z_{1-\alpha}\big)$ of
Proposition~\ref{prop:sharp}. Worst-case power stays  nominal across the dead zone
$|\xi|\le 2\bar b$. Under any single bias the
undetected window has half-width $\bar b$, and the dead zone is the union of these windows. \textbf{(b)} Fix the truth at $\beta_0$ and ignore sampling error. Under the null $b_0$,
honesty places the truth somewhere in $b_0\pm\bar b$, so no value in
that band can be rejected. With $\beta_0$ at the lower end of the band (bias $+\bar b$) the
non-rejected values reach $\beta_0+2\bar b$, and with $\beta_0$ at the upper end (bias $-\bar b$)
they reach $\beta_0-2\bar b$. The bias is untestable, so both positions must be defended under the
honest promise, and their union is the dead zone.}
\label{fig:deadzone}
\end{figure}

\begin{example}[Conservative to weak factor models]\label{ex:WeakFactors}
    Consider the panel data linear regression model setting studied in \cite{armstrong2022robust}. 
    Take data generated as 
    \begin{align*}
        y_{it} = x_{it}'\beta +  \Gamma_{it} + \varepsilon_{it}, &&
        i\in \{1,\dots, N\}, t\in \{1,\dots, T\}.
    \end{align*}
    The term $\Gamma_{it}$ may be correlated with $x_{it}$. 
    For a given estimate $\hat\beta_{ba}$, an upper bound on absolute bias, $\bar B:= \sup_\theta|\hat B|$, 
    is defined by $\theta = R_{weak}$, i.e. the number of factors that are declared weak. 
    The proposed confidence intervals in \cite{armstrong2022robust}, which take the fixed-length bias-aware form of \citet{donoho1994statistical}, are widened to, 
    \begin{align*}
        \{\hat\beta_{ba} \pm [\bar B + z_{1-\alpha/2}\cdot se(\hat\beta_{ba})]\}
    \end{align*}
    where $z_{1-\alpha/2}$ is the $1-\alpha/2$ quantile of the Gaussian distribution. 
    When weak factors are truly present, $\{\bar{B}, \hat B\} = O_p(\max\{N,T\}/n)$, where $n = NT$. 
    However, when strong factors are truly present, $\bar{B} = O_p(\max\{N,T\}/n)$, whilst it can be that $\hat B = o_p(1/\sqrt{n})$.\footnote{Refer to \cite{Bai2009} Theorems 2 and 4 for parametric rate convergence of $\hat\beta$ under strong factors. }  
\end{example}

\begin{figure}[!ht]
    \centering
    \includegraphics[width=0.32\linewidth, height = 5cm]{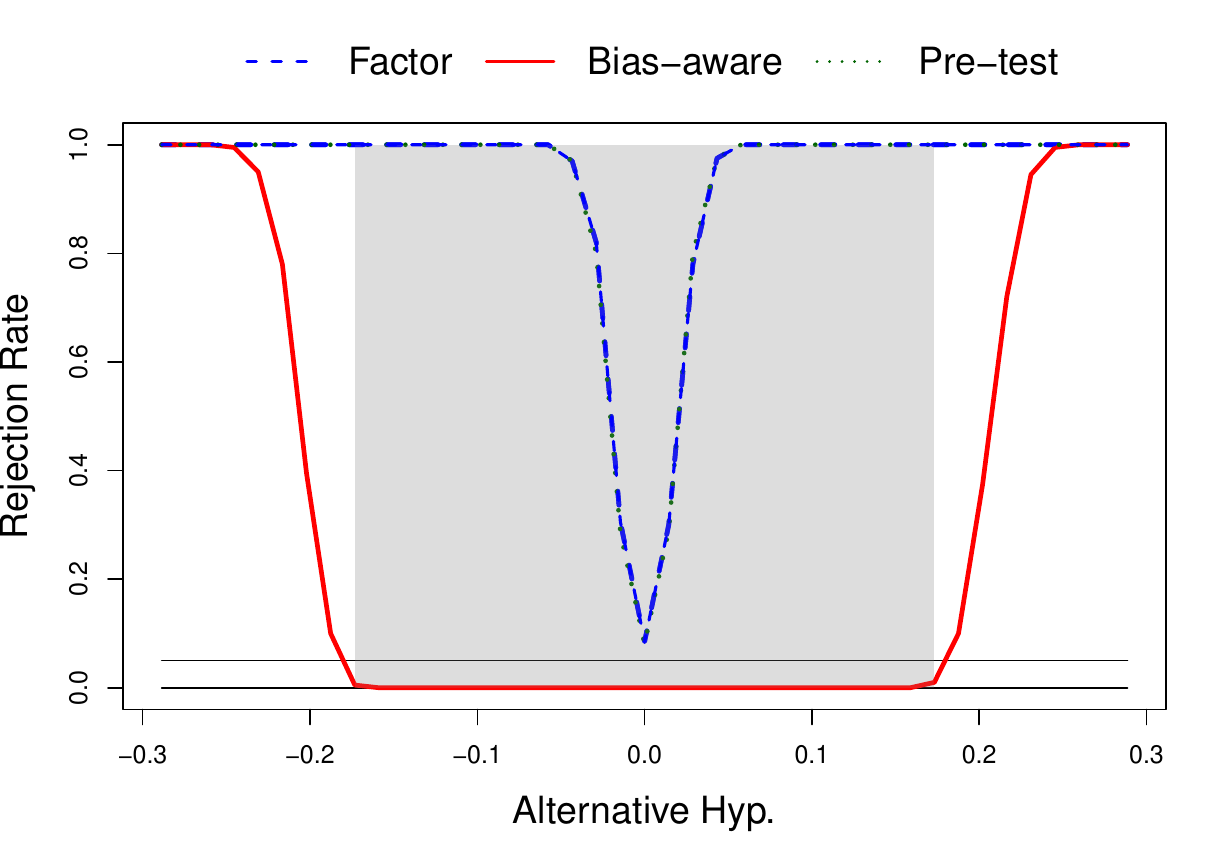}
    \includegraphics[width=0.32\linewidth, height = 5cm]{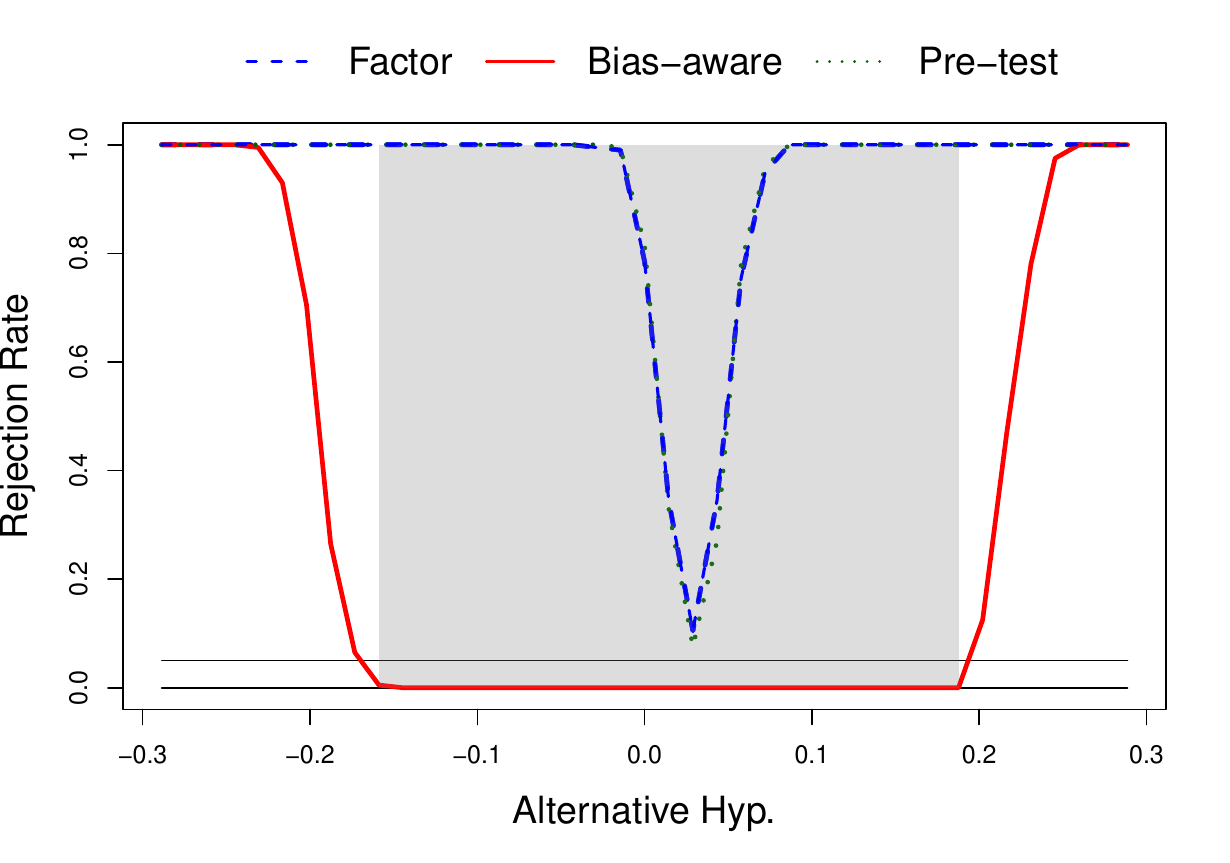}
    \includegraphics[width=0.32\linewidth, height = 5cm]{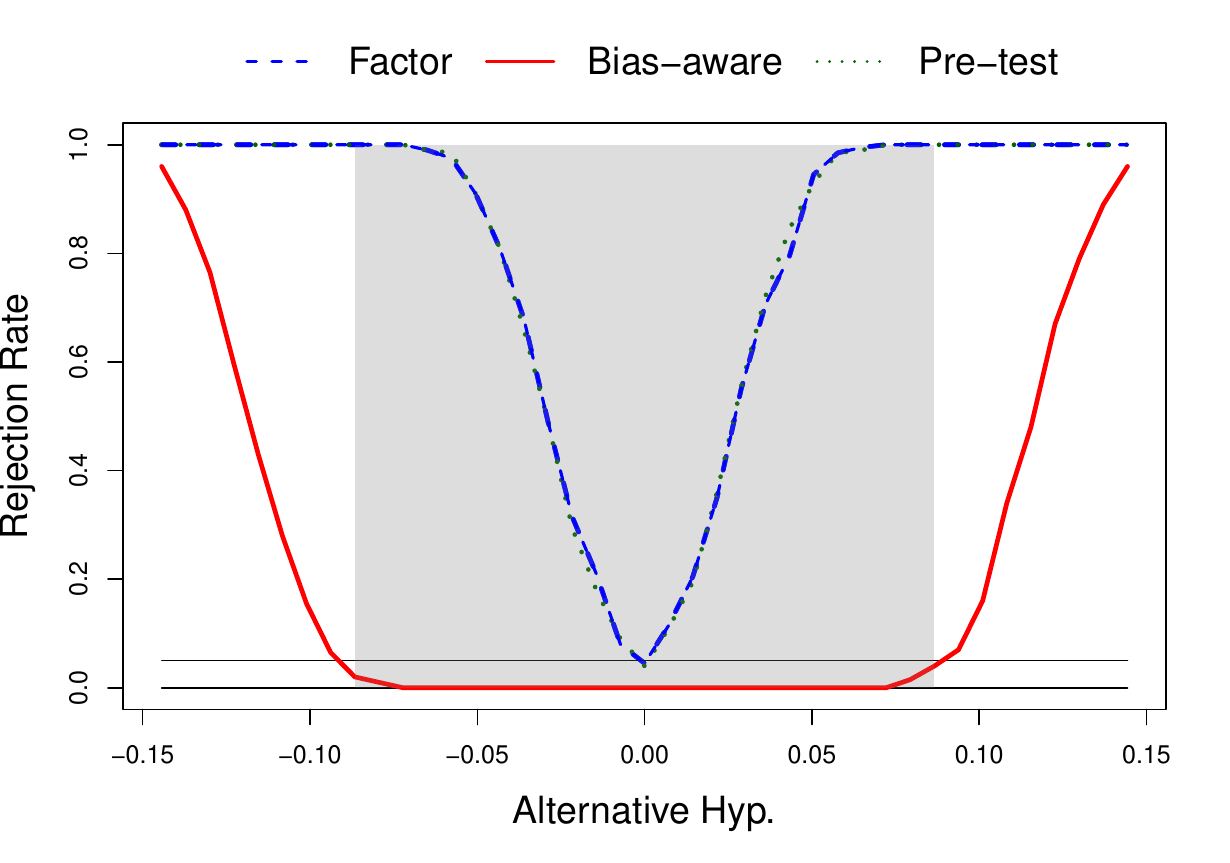}
    \caption{
    Rejection rates against $n^{-1/2}$ local alternatives.
    Left to right: No vs. Weak vs. Strong factor. 
    Blue dashed line: \cite{Bai2009};
    Red solid line: Bias-aware;
    Green dotted line: pre-test (eigenvalue ratio; least squares when $\hat R_{strong}=0$).
    Factor model over-rejects the truth with weak factors, nominally sized with no or strong factors. Its non-rejection window is tight in all three
    panels. The pre-test tracks the factor model in every panel. Detection fails exactly
    where it is needed. The bias-aware declaration follows the feasible rule (factors not
    certifiably strong are declared weak), so the dead zone is wider in the no and weak
    panels than in the strong panel.
    }
    \label{fig:IntroSimulation}
\end{figure}

Rejection rates for Example~\ref{ex:WeakFactors}, the debiased estimator in \cite{armstrong2022robust}, in Figure~\ref{fig:IntroSimulation} are taken from simulations in Section~\ref{sect:simFactor}. 
The CIs conservative to weak factors have absolute upper bounds on bias of order $\bar{B} = O_p(\max\{N,T\}/n)$, where $n = NT$.
Take $T = c\cdot N^a$ for $c\in (0,1)$ and $a \in [0,1]$, i.e. $N = \max\{N,T\}$, such that $\sqrt{n}$ multiplied by bias is of order $\sqrt{n}\bar{B} = O_p(N^{(1 -  a)/2})$. 
Then, for $a = 1$, i.e. $N$ and $T$ proportional to each other, scaled sup bias is $\sqrt{n}\bar B = O_p(1)$, i.e. may be bounded by a constant but is not growing asymptotically. 
This leads to finite sample power loss, with CIs converging at the parametric rate. 
However, when $a \neq 1$, scaled asymptotic sup bias is asymptotically divergent, and the bias-aware confidence intervals retain no local asymptotic power when local power is defined at the $\sqrt{n}$ rate. 
This power loss is without gain when true bias is asymptotically negligible, e.g. with no factor at all or under a strong true factor model. 

The same structure underlies two settings common in applied work.

\begin{example}[Conservative to parallel trend violations]\label{ex:DiD}
    Take the honest difference-in-differences setting of \cite{rambachan2023more}.
    An event study delivers pre- and post-treatment coefficients, with the post-period coefficient $\hat\beta_{post}$ targeting the treatment effect $\beta$,
    \begin{align*}
        \hat\beta_{post} - \beta = \hat B + O_p(1/\sqrt{n}),
    \end{align*}
    where $\hat B$ is the post-period violation of parallel trends.
    The analyst supplies $\theta = M$, a bound on the second differences of the differential trend (the class $\Delta^{SD}(M)$), giving, once the pre-trend is extrapolated away as the honest procedure does, the worst-case bias $\bar B := \sup_\theta |\hat B|$ of exact order $M$, constant in $n$.
    The honest interval widens to $\{\hat\beta_{post} \pm [\bar B + z_{1-\alpha/2}\cdot se(\hat\beta_{post})]\}$.
    When parallel trends are violated up to the bound, $\{\bar B, \hat B\}$ are of exact order $M$.
    However, when parallel trends hold, $\bar B$ is of exact order $M$ whilst $\hat B = o_p(1/\sqrt{n})$ ($\hat B = 0$ here).
\end{example}

\begin{example}[Conservative to invalid instruments]\label{ex:PlausExog}
    Consider the plausibly-exogenous IV setting of \cite{conley2012plausibly}, with structural outcome $y_i = \beta x_i + \gamma z_i + \varepsilon_i$, instrument $z_i$, and first stage $x_i = \pi z_i + v_i$.
    The IV estimator that imposes exogeneity satisfies
    \begin{align*}
        \hat\beta_{IV} - \beta = \hat B + O_p(1/\sqrt{n}), \qquad \hat B \xrightarrow{p} \gamma/\pi,
    \end{align*}
    the asymptotic bias from a true exclusion violation $\gamma$.
    The analyst supplies $\theta = \bar\gamma$, a bound on $|\gamma|$, giving the worst-case bias $\bar B := \sup_\theta |\hat B|$ of exact order $\bar\gamma/\pi$, constant in $n$.
    The union-of-intervals procedure of \cite{conley2012plausibly} widens to $\{\hat\beta_{IV} \pm [\bar B + z_{1-\alpha/2}\cdot se(\hat\beta_{IV})]\}$.
    When the instrument is invalid up to the bound, $\{\bar B, \hat B\}$ are of order $\bar\gamma/\pi$.
    However, when the instrument is valid ($\gamma = 0$), $\bar B$ is of exact order $\bar\gamma/\pi$ whilst $\hat B = o_p(1/\sqrt{n})$ ($\hat B = 0$ here).
\end{example}

\subsection*{Related literature}

In nonparametric regression the user-specified bound is a smoothness or magnitude constant. \citet{donoho1994statistical} characterises fixed-length intervals for linear functionals through the modulus of continuity over a convex parameter space, \citet{li1989honest} shows honest intervals must account for smoothing bias, and \citet{armstrong2018optimal,armstrong2020simple} replace the usual critical value with one incorporating the worst-case bias. Extensions include regression discontinuity with a discrete running variable \citep{kolesar2018inference}, finite-sample minimax regression discontinuity under a bound on the second derivative \citep{imbens2019optimized}, bounded control coefficients \citep{armstrong2020biasaware}, and fuzzy regression discontinuity \citep{noack2024bias}.

That the bound cannot be learned from the data is fundamental. Within standard smoothness classes an honest interval cannot adapt its length to the unknown regularity \citep{low1997nonparametric,cailow2004adaptation}, so the worst case pins the length, echoing \citet{bahadur1956nonexistence}. In the bounded normal mean, the subproblem at the core of the limit experiment below, the minimax half-length of a fixed-length $1-\alpha$ interval equals the bias bound whenever the bound does not exceed $z_{1-\alpha}$ noise units \citep[eq.~(9)]{donoho1994statistical}. This is a length statement, and its power counterpart, at twice the scale and uniform over all honest procedures, is developed in Section~\ref{sect:main}.

The paper also relates to the literature on model misspecification. \citet{conley2012plausibly} relax the exclusion restriction by a user-chosen amount, \citet{armstrong2021sensitivity} conduct sensitivity analysis under bounded local misspecification of moment conditions, and \citet{andrews2017measuring} and \citet{bonhomme2022minimizing} study the sensitivity of estimates to misspecified moments. Robust constructions of this kind continue to appear, for instance for weighted estimands under an inferred bound on treatment-effect heterogeneity \citep{vilfort2026robust}, and the power accounting here applies to such settings, where reporting the dead zone would be a natural complement to their coverage guarantees.

The duality between interval length and the power of the dual test is classical \citep{pratt1961length}, and several of these works state the size--power trade-off explicitly \citep[e.g.][]{li1989honest,low1997nonparametric,noack2024bias}. \citet{rambachan2023more}, whose approach is motivated in part by the low power of pre-trend tests \citep{roth2022pretest}, characterise the power of their tests as a function of the assumed bound. 
What \citet{rambachan2023more} miss is the full reach of twice the bias bound that the honest promise actually guarantees coverage over, addressed here with the dead zone. 
The contribution is to make the trade-off precise when the relevant length is not a fixed multiple of the standard error but a user-specified bias bound whose asymptotic order may exceed the estimator's convergence rate.

\section{Main result}\label{sect:main}

This section begins with a simple and instructive result on the power of the dual test (the test that rejects a value when the interval excludes it) under bias-aware intervals.
This establishes the breadth of the paper's contribution.
Section~\ref{sect:minimaxpowerLossHonest} establishes that the power loss is an inherent feature of the honest requirement, regardless of the particular procedure used to meet it.
That is, the power loss is unavoidable once honesty under undetectable bias is deemed strictly necessary. 

Assume that $\bar B := \sup_\theta |\hat B|$ is bounded. Notation $\Theta_p $ implies an exact asymptotic bound, namely $A_n = \Theta_p(B_n) \iff A_n = O_p(B_n)$ and $A_n^{-1} = O_p(B_n^{-1})$. 

\subsection{Power loss of bias-aware intervals}\label{sect:powerLossBACI}

\begin{assumption}\label{ass:AnalystEst}
Assume the analyst possesses $\hat{\beta}$, and $\bar B:= \sup_\theta |\hat B|$ such that, 
\begin{align*}
    \hat\beta - \beta = \hat{B} + O_p(n^{-\epsilon}). 
\end{align*}  
    Moreover, assume the analyst possesses $se(\hat\beta) = \Theta_p(n^{-\epsilon})$.
\end{assumption}

\begin{assumption}\label{ass:BiasBound}
    For sample size $n$, and $\bar B := \sup_\theta |\hat B|$, there is
        $\bar B  = \Theta (n^{-\delta})$
    for $\delta \in [0, \infty)$. 
\end{assumption}
Assumption~\ref{ass:BiasBound} upper bounds the supremum over asymptotic bias to be constant, and admits settings with non-negligible asymptotic bias if $\delta < \epsilon$. 

Local alternative hypotheses are defined per \cite{newey1994large} as follows.
\begin{definition}\label{def:LocalAlt}
    For $\rho \in [0, \epsilon]$ the null and the local alternatives, indexed by the set $\Xi$, are
    \begin{align*}
        H_0: \beta = \beta_0 
        &&
        \textrm{vs. }
        &&
        H_{\rho}: \beta = \beta_0 + \xi\cdot n^{-\rho},
        &&
        \xi\in\Xi
        .
    \end{align*}
\end{definition}

Throughout the power analysis $\beta_0$ is the null value and the data are generated under it
($\beta=\beta_0$), i.e. the null is correctly specified. It is then asked how often nearby false values are rejected,
the tested null shifting from $\beta_0$ to the local alternative $A_n=\beta_0+\xi\,n^{-\rho}$.

Firstly note, the subscript
in $\Pr_{\beta_0}$ names the value at which the data are generated. Secondly, the only random object in the event $\{A_n\notin\mathcal B_n\}$ is the
interval $\mathcal B_n$, while the tested value $A_n$ is a deterministic
sequence. Third, an interval \emph{excludes} $A_n$ when $A_n$ falls outside it, which is
precisely the event that the dual test rejects $A_n$. The local power is therefore
$\Pr_{\beta_0}\{A_n\notin\mathcal B_n\}$, and zero local power means the interval covers these
false values with probability approaching one. The dead zone will turn out to have the same
width about any hypothesised value, so the applications of Section~\ref{sect:applications}
centre it at the economically natural null of no effect. 
Fixed alternative hypotheses are admitted in Definition~\ref{def:LocalAlt} by setting $\rho = 0$.

\begin{theorem}\label{thm:LocalCI}
    Under Assumptions~\ref{ass:AnalystEst} and \ref{ass:BiasBound}, and $\rho \in [0, \epsilon]$ with $\delta<\rho$, 
    \begin{align*}
        n^{\rho}\cdot(\bar B + z_{1-\alpha/2}\cdot se(\hat{\beta})) = \Theta_p(n^{\rho-\delta}) + O_p(1) \to \infty.
    \end{align*}
    Moreover, for $\rho \in [0, \epsilon]$ and $\delta\geq \rho$,
    \begin{align*}
        n^{\rho}\cdot(\bar B + z_{1-\alpha/2}\cdot se(\hat{\beta})) = O_p(1) .
    \end{align*}
\end{theorem}
\begin{proof}
    Theorem~\ref{thm:LocalCI} is immediate under the stated assumptions. 
\end{proof}

Theorem~\ref{thm:LocalCI} states an immediate implication of scaling up confidence intervals at a rate faster than the worst-case bias shrinks: the intervals diverge.
Its hypothesis is satisfied by a range of honest procedures in current use \citep{armstrong2022robust,rambachan2023more,conley2012plausibly}, so the power loss is inherited wholesale rather than by any one construction. 
Section~\ref{sect:minimaxpowerLossHonest} shows that the loss it induces is intrinsic to honesty under user defined biases.

The results of this subsection share one mechanism. After rescaling,
the local power function of any interval is the non-coverage function of its limiting rescaled
form. This is the organising claim of the introduction, namely that the limit of the rescaled
width determines local power.

\begin{lemma}[Width--power duality]\label{lem:containment}
Let $\mathcal B_n=[\,l_n,u_n\,]$ be any interval and suppose that, with data generated at
the truth $\beta = \beta_0$, the rescaled endpoints converge jointly in distribution,
\begin{align*}
\big(\,n^{\rho}(l_n-\beta_0),\ n^{\rho}(u_n-\beta_0)\,\big)\ \xrightarrow{d}\ (L,U),
\end{align*}
with values in extended reals. Then, for $\xi$ not an atom of endpoints,
$\Pr\{L=\xi\}=\Pr\{U=\xi\}=0$,
\begin{align*}
\lim_{n\to\infty}\Pr\left\{A_n\notin\mathcal B_n\right\}\ =\ \Pr\left\{\xi\notin[L,U]\right\},
\qquad A_n=\beta_0+\xi\,n^{-\rho}.
\end{align*}
In particular every $\xi$ with $\Pr\{L\le\xi\le U\}\ge1-\alpha$ lies in the dead zone of
$\mathcal B_n$.
\end{lemma}
\begin{proof}
Write $L_n=n^{\rho}(l_n-\beta_0)$ and $U_n=n^{\rho}(u_n-\beta_0)$. Since $l_n\le u_n$, the
exclusion event is the disjoint union
$\{A_n\notin\mathcal B_n\}=\{L_n>\xi\}\cup\{U_n<\xi\}$. Convergence in distribution of the
endpoints means their distribution functions converge at every continuity point of the limits,
and $\Pr\{L=\xi\}=\Pr\{U=\xi\}=0$ makes $\xi$ such a point for both, so
$\Pr\{L_n>\xi\}\to\Pr\{L>\xi\}$ and $\Pr\{U_n<\xi\}\to\Pr\{U<\xi\}$. The limit events can occur
together only if $U<\xi<L$, which $L\le U$ rules out (equality included), so adding the two
limits gives $\Pr\{\xi\notin[L,U]\}$. Only the marginal distributions of the
endpoints enter.
\hfill
\end{proof}

Theorems~\ref{thm:LocalPower}--\ref{thm:nondeg} below are instances, differing only in the
limit $(L,U)$ of the bias-aware interval. The limit is $(-\infty,+\infty)$ in the bound-dominated regime
$\delta<\rho$ (zero power), the constants $\big((c-1)\bar b,(c+1)\bar b\big)$ at the matched
rate with negligible noise, and Gaussian endpoints when noise
survives. Proposition~\ref{prop:partialID} is the same statement with the identified set as
the limit.

Theorem~\ref{thm:LocalPower} states that whenever $\hat B = o_p(\bar B)$, nearby false values that are asymptotically small with respect to the bias bound, $\bar B$, are not rejected with probability approaching one (wpa1).
Whilst the set of local alternative scalings, $\Xi_n$, can be unbounded in the limit, the rate this can grow is potentially very slow. Thus, it is easier to interpret this as a compact set. 

\begin{theorem}\label{thm:LocalPower}
     Make Assumptions~\ref{ass:AnalystEst} and \ref{ass:BiasBound}, and use Definition~\ref{def:LocalAlt}, with $\beta = \beta_0$, where the alternative set may grow with $n$, namely $\xi\in\Xi_n\subset\mathbb{R}$ with $\sup|\Xi_n| := \sup_{\xi\in\Xi_n}|\xi| = o(n^{\rho - \delta})$.
    Let $A_n = \beta_0 + \xi\cdot n^{-\rho}$, and $\mathcal{B}_n = \left[\hat\beta  \pm (\bar B + z_{1-\alpha/2}\cdot se(\hat{\beta}))\right]$ with $\rho\in[0,\epsilon)$ and $\delta < \rho$. 
    Then, when $\hat B = o_p(\bar B)$, for any $\xi\in \Xi_n$,
    \begin{align*}
        \lim_{n\to\infty}\Pr_{\beta_0}\left\{A_n\in \mathcal{B}_n\right\}  = 1.
    \end{align*}
\end{theorem}

\begin{proof}
By definition, $A_n \in \mathcal{B}_n$ if and only if
\begin{align*}
    \hat\beta - \left(\bar B + z_{1-\alpha/2}\cdot se(\hat{\beta})\right)
    \leq A_n \leq
    \hat\beta + \left(\bar B + z_{1-\alpha/2}\cdot se(\hat{\beta})\right).
\end{align*}
Substituting $A_n = \beta_0 + \xi\cdot n^{-\rho}$ and $\hat\beta - \beta_0 = \hat B + O_p(n^{-\epsilon})$ from
Assumption~\ref{ass:AnalystEst}, subtracting $\beta_0$, and multiplying through by $n^{\rho}$, this is
equivalent to $L_n \leq \xi \leq U_n$, where
\begin{align}
    L_n &= n^{\rho}(\hat B - \bar B) - z_{1-\alpha/2}\cdot n^{\rho} se(\hat{\beta}) + O_p(n^{\rho - \epsilon}),
    \nonumber
    \\
    \label{eqn:containment}
    U_n &= n^{\rho}(\hat B + \bar B) + z_{1-\alpha/2}\cdot n^{\rho} se(\hat{\beta}) + O_p(n^{\rho - \epsilon}).
\end{align}
From $\delta < \rho < \epsilon$, these terms are, by Assumptions~\ref{ass:AnalystEst}
and \ref{ass:BiasBound},
\begin{align*}
    n^{\rho}\bar B = \Theta(n^{\rho - \delta}) \to \infty,
    \qquad
    n^{\rho}\hat B = o_p\!\left(n^{\rho}\bar B\right),
    \qquad
    n^{\rho} se(\hat{\beta}) = \Theta_p(n^{\rho - \epsilon}) = o_p(1),
\end{align*}
since $\hat B = o_p(\bar B)$. Hence
$n^{\rho}(\hat B \pm \bar B) = \pm\, n^{\rho}\bar B\,(1 + o_p(1)) = \pm\,\Theta_p(n^{\rho - \delta})$, so that
\begin{align*}
    U_n = \Theta_p(n^{\rho - \delta}) \to +\infty,
    \qquad
    L_n = -\Theta_p(n^{\rho - \delta}) \to -\infty.
\end{align*}
Every $\xi \in \Xi_n$ satisfies $|\xi| \leq \sup|\Xi_n| = o(n^{\rho - \delta})$, while $U_n$ and $-L_n$
are of exact order $n^{\rho - \delta}$. Therefore the interval $[L_n, U_n]$ contains
$[-\sup|\Xi_n|, \sup|\Xi_n|] \supseteq \{\xi\}$ with probability approaching one, and event
\eqref{eqn:containment} holds wpa1. Consequently,
\begin{align*}
    \lim_{n\to\infty} \Pr_{\beta_0}\left\{ A_n \in \mathcal{B}_n \right\} = 1.
\end{align*}
\hfill
\end{proof}

Theorem~\ref{thm:LocalPower} shows zero local power under a possibly slack bound, $\hat B = o_p(\bar B)$. 
In the sequel, Theorem~\ref{thm:LocalCI2} shows that this power loss persists even when the sign of the bias is correct and the bound is tight, provided the alternatives shrink faster than the bias bound, i.e. $\delta<\rho$.
Knowing the sign of bias is of course outside the capabilities of bias-aware confidence intervals, since then a simple debias could be performed. It is nonetheless a useful exercise to study power properties had these been pinned down. 
 
\begin{theorem}\label{thm:LocalCI2}
    Make the assumptions of Theorem~\ref{thm:LocalPower}, with $\hat B = o_p(\bar B)$ replaced by,
    \begin{align*}
        \hat B = c\cdot\bar B + o_p(n^{-\epsilon}), \qquad c\in[-1,1].
    \end{align*}
    Assume $\delta<\rho$ and $\Xi$ compact. Then, for every fixed $\xi\in\Xi$ with $\xi\neq 0$,
    \begin{align*}
        \lim_{n\to\infty}\Pr_{\beta_0}\left\{A_n\in\mathcal{B}_n\right\}
        =
        \begin{cases}
            1, & c\in(-1,1),
            \\
            \mathbbm{1}\{\operatorname{sign}(\xi)=c\}, & c\in\{-1,1\}.
        \end{cases}
    \end{align*}
\end{theorem}
\begin{proof}
    By the proof of Theorem~\ref{thm:LocalPower}, $A_n\in\mathcal{B}_n$ if and only if $\xi\in[L_n,U_n]$, with $L_n,U_n$ as in \eqref{eqn:containment}. Substituting $\hat B = c\bar B + o_p(n^{-\epsilon})$, and using $n^{\rho}\bar B=\Theta(n^{\rho-\delta})\to\infty$ and $n^{\rho}se(\hat\beta)=\Theta_p(n^{\rho-\epsilon})=o_p(1)$ (since $\rho<\epsilon$),
    \begin{align*}
        L_n = (c-1)\,n^{\rho}\bar B + o_p(n^{\rho}\bar B),
        \qquad
        U_n = (c+1)\,n^{\rho}\bar B + o_p(n^{\rho}\bar B).
    \end{align*}
    If $c\in(-1,1)$ then $c-1<0<c+1$, so $L_n\to-\infty$ and $U_n\to+\infty$, and $[L_n,U_n]$ contains any fixed $\xi$ wpa1. If $c=1$ then $U_n\to+\infty$ while $(c-1)n^{\rho}\bar B=0$, leaving $L_n=-z_{1-\alpha/2}\,n^{\rho}se(\hat\beta)+O_p(n^{\rho-\epsilon})=o_p(1)$, so $[L_n,U_n]$ covers a fixed $\xi$ wpa1 if and only if $\xi>0$. Symmetrically, if $c=-1$ then $[L_n,U_n]$ covers $\xi$ wpa1 if and only if $\xi<0$. In both boundary cases coverage occurs wpa1 exactly when $\operatorname{sign}(\xi)=c$.
    \hfill
\end{proof}
 
Theorem~\ref{thm:LocalCI2} isolates the role of the bound's tightness. 
When $\delta<\rho$ the scaled bound $n^{\rho}\bar B$ diverges, so even a bound that exactly equals the bias in the limit with known sign cannot restore two-sided power. 
The test is asymptotically powerless against local alternatives on the same side as the bias, and detects only opposite-signed alternatives. 
For any $c\in(-1,1)$ all local power is lost, exactly as in Theorem~\ref{thm:LocalPower}.
 
The matched-rate case $\delta=\rho$, where the alternative shrinks at precisely the rate of the bias bound, is the knife edge between the bound-dominated regime above and the standard regime $\delta>\rho$. Theorem~\ref{thm:nondeg} pins down the resulting non-degenerate power loss. 
 
\begin{theorem}\label{thm:nondeg}
    Make Assumptions~\ref{ass:AnalystEst} and \ref{ass:BiasBound}, and use Definition~\ref{def:LocalAlt} with $\beta=\beta_0$. Take $\delta=\rho\in[0,\epsilon]$, bias $\hat B = c\cdot\bar B + o_p(n^{-\epsilon})$ with $c\in[-1,1]$, and suppose $n^{\delta}\bar B\xrightarrow{p}\bar b\in(0,\infty)$. Fix $\xi\in\Xi$.
    \begin{enumerate}[(i)]
    \item If $\rho<\epsilon$, then $n^{\rho}se(\hat\beta)=o_p(1)$ and, for every $\xi\notin\{(c-1)\bar b,\,(c+1)\bar b\}$, the local power is degenerate,
    \begin{align*}
        \lim_{n\to\infty}\Pr\left\{A_n\notin\mathcal{B}_n\right\}
        =
        \mathbbm{1}\big\{\,\xi\notin[\,(c-1)\bar b,\;(c+1)\bar b\,]\,\big\}.
    \end{align*}
    \item If $\rho=\epsilon$,  assuming ${n}^\epsilon\,(\hat\beta-\beta_0-\hat B)\xrightarrow{d}\mathcal{N}(0,\sigma^2)$ and ${n}^\epsilon\,se(\hat\beta)\xrightarrow{p}\sigma$, local power is
    \begin{align*}
        \lim_{n\to\infty}\Pr\left\{A_n\notin\mathcal{B}_n\right\}
        =
        \Pr\big\{\,|\,G+c\bar b-\xi\,| > \bar b + z_{1-\alpha/2}\,\sigma\,\big\},
        \qquad
        G\sim\mathcal{N}(0,\sigma^2).
    \end{align*}
    \end{enumerate}
\end{theorem}
\begin{proof}
    By the proof of Theorem~\ref{thm:LocalPower}, $A_n\in\mathcal{B}_n$ if and only if $\xi\in[L_n,U_n]$, where $[L_n,U_n]$ has centre $n^{\rho}\hat B + O_p(n^{\rho-\epsilon})$ and half-width $n^{\rho}\bar B + z_{1-\alpha/2}\,n^{\rho}se(\hat\beta)$. With $\delta=\rho$, $n^{\rho}\bar B = n^{\delta}\bar B\xrightarrow{p}\bar b$ and $n^{\rho}\hat B = c\,n^{\delta}\bar B + o_p(n^{\rho-\epsilon})\xrightarrow{p}c\bar b$.
 
    (i) For $\rho<\epsilon$, $n^{\rho}se(\hat\beta)=\Theta_p(n^{\rho-\epsilon})=o_p(1)$ and the $O_p(n^{\rho-\epsilon})$ term vanishes, so $L_n\xrightarrow{p}(c-1)\bar b$ and $U_n\xrightarrow{p}(c+1)\bar b$. For fixed $\xi$ away from the endpoints, $\Pr\{\xi\in[L_n,U_n]\}\to\mathbbm{1}\{\xi\in[(c-1)\bar b,(c+1)\bar b]\}$, which gives the stated rejection probability.
 
    (ii) For $\rho=\epsilon$, ${n}^\epsilon(\hat\beta-A_n)={n}^\epsilon\hat B-\xi+{n}^\epsilon(\hat\beta-\beta_0-\hat B)\xrightarrow{d}c\bar b-\xi+G$, while the half-width ${n}^\epsilon\bar B + z_{1-\alpha/2}{n}^\epsilon se(\hat\beta)\xrightarrow{p}\bar b+z_{1-\alpha/2}\sigma$. Acceptance is the event $|{n}^\epsilon(\hat\beta-A_n)|\le$ half-width, and the continuous mapping theorem yields the stated limit.
    \hfill
\end{proof}
 
Theorem~\ref{thm:nondeg} displays the cost of conservatism.
When $\rho=\delta$ the bound neither washes out the alternative ($\delta<\rho$) nor becomes negligible ($\delta>\rho$), and contributes a constant $2\bar b$ to width. 
When $\rho<\epsilon$, the sampling error is of smaller order, so the limiting power is a step function.
Values in $[(c-1)\bar b,(c+1)\bar b]$ are never rejected, while every value outside is rejected wpa1.
This window has half-width $\bar b$ and depends on the realised bias $c$, which is untestable.
Its union over $c\in[-1,1]$ is $[-2\bar b,2\bar b]$, the worst-case dead zone of
Section~\ref{sect:minimaxpowerLossHonest}.

Conservatism thus leads to several scenarios. 
The dead zone has half-width exactly twice the limiting bound, against which the test has no
power. Setting $\bar b=0$ recovers the oracle that rejects every $\xi\neq0$. At the rate $\rho=\epsilon$ the sampling error survives, and the power curve is smooth.
It is the usual Gaussian power function, with the rejection threshold inflated from $z_{1-\alpha/2}\sigma$ to $\bar b+z_{1-\alpha/2}\sigma$ and recentred by the bias $c\bar b$. The added $\bar b$ is the exact, finite power cost of bias-aware conservatism.

\begin{remark}[The cost in required sample size]\label{rem:pitman}
The dead zone can be restated in terms of required sample size, the comparison underlying
Pitman relative efficiency
\citep[Ch.~14]{vandervaart1998asymptotic}, namely the number of observations needed to reach a
prescribed level and power. Fix a level $\alpha$, a target power $\gamma$, an effect
$\Delta>0$ and a fixed bias bound $\bar B$ (the $\delta=\rho=0$ case), with
$se(\hat\beta)=\sigma n^{-\epsilon}$ and $z_\gamma=\Phi^{-1}(\gamma)$. Setting the one-sided
honest worst-case power $\Phi\big(n^{\epsilon}(\Delta-2\bar B)/\sigma-z_{1-\alpha}\big)$ equal
to $\gamma$ gives the minimal sample sizes
\begin{align*}
n_{\mathrm{honest}}\approx\Big[\frac{(z_{1-\alpha}+z_{\gamma})\,\sigma}{\Delta-2\bar B}\Big]^{1/\epsilon},
&&
n_{\mathrm{oracle}}\approx\Big[\frac{(z_{1-\alpha}+z_{\gamma})\,\sigma}{\Delta}\Big]^{1/\epsilon},
&&
\frac{n_{\mathrm{honest}}}{n_{\mathrm{oracle}}}\;\longrightarrow\;
\Big(\frac{\Delta}{\Delta-2\bar B}\Big)^{1/\epsilon},
\end{align*}
the square of $\Delta/(\Delta-2\bar B)$ at the parametric rate $\epsilon=1/2$. The ratio
diverges as $\Delta\downarrow2\bar B$, and for $\Delta\le2\bar B$ no sample size delivers
power $\gamma>\alpha$. The Pitman efficiency of honest inference relative to the oracle is
zero inside the dead zone. Unlike the classical comparison
\citep[Thm.~14.19]{vandervaart1998asymptotic}, where the sample-size ratio is a single number
independent of $(\alpha,\gamma)$ and the alternative, here it depends on $\Delta$, because
honesty translates the power curve rather than rescaling its slope, and that dependence is
the substantive finding. The burden of proof grows without bound as the effect approaches the edge of
the dead zone.
\end{remark}

\subsection{A minimax lower bound for honest inference}\label{sect:minimaxpowerLossHonest}

The results above concern a single procedure, i.e. the bias-aware interval
$\mathcal{B}_n$ and the test it induces. 
It is shown here that every test that
is honest faces the same three-regime power behaviour, so the loss
documented in Section~\ref{sect:powerLossBACI} is a property of the inference problem rather than just the
bias-aware confidence interval widening. The matched-rate ceiling takes the form
$\alpha+\tau$, with $\tau$ the detectability of the confounding defined below. In the exactly
confounded settings of Lemma~\ref{lem:exactconf}, $\tau\equiv0$ and the ceiling is the flat
dead zone itself.

Let $P_{\beta,\theta}^{(n)}$ denote the data-generating process, that is, the joint
distribution of the data, with target coefficient
$\beta$ and untestable nuisance $\theta\in\mathcal{C}_n$, and let
$\mathcal{P}_n = \{P_{\beta,\theta}^{(n)}: \beta\in\mathbb{R},\, \theta\in\mathcal{C}_n\}$ be
the class over which coverage is required. A test is a measurable
$\phi_n\in\{0,1\}$, the zero-one rejection indicator generated by an underlying test
statistic, critical value, and rejection rule, rejecting $H_0:\beta=\beta_0$ when $\phi_n=1$.

\begin{definition}\label{def:honest}
    A test sequence $\{\phi_n\}$ is \emph{asymptotically honest at level $\alpha$} over
    $\{\mathcal{P}_n\}$ if
    \begin{align*}
        \limsup_{n\to\infty}\ \sup_{\theta\in\mathcal{C}_n}\ \Pr_{P_{\beta_0,\theta}^{(n)}}\left\{\phi_n = 1\right\}\ \leq\ \alpha.
    \end{align*}
\end{definition}

Definition~\ref{def:honest} is the hypothesis test counterpart of uniform coverage.
Recall that an interval $\mathcal{B}_n$ is honest over $\{\mathcal{P}_n\}$
at level $1-\alpha$ if it covers each process's own parameter value uniformly,
\begin{align}\label{eqn:standardHonest}
    \liminf_{n\to\infty}\ \inf_{\beta\in\mathbb{R},\ \theta\in\mathcal{C}_n}\ \Pr_{P_{\beta,\theta}^{(n)}}\left\{\beta\in\mathcal{B}_n\right\}\ \geq\ 1-\alpha.
\end{align}
Pairing $\mathcal{B}_n$ with its dual tests $\phi_n(b):=\mathbbm{1}\{b\notin\mathcal{B}_n\}$,
covering the true value is non-rejection of the true value.
For every process,
$\Pr_{P_{\beta,\theta}}\{\beta\in\mathcal{B}_n\}=1-\Pr_{P_{\beta,\theta}}\{\phi_n(\beta)=1\}$.
Taking the infimum and using
$\liminf_n(1-s_n)=1-\limsup_n s_n$, condition \eqref{eqn:standardHonest} is equivalent to uniform
asymptotic size control of the dual tests over the whole class,
\begin{align}\label{eqn:dualHonest}
    \limsup_{n\to\infty}\ \sup_{\beta\in\mathbb{R},\ \theta\in\mathcal{C}_n}\ \Pr_{P_{\beta,\theta}^{(n)}}\left\{\phi_n(\beta)=1\right\}\ \leq\ \alpha.
\end{align}
Retaining only the term $\beta=\beta_0$ in \eqref{eqn:dualHonest}, which can only lower the
supremum, yields Definition~\ref{def:honest} for the dual test
$\phi_n(\beta_0)=\mathbbm{1}\{\beta_0\notin\mathcal{B}_n\}$, and equally the size control of the
dual test $\phi_n(b)=\mathbbm{1}\{b\notin\mathcal{B}_n\}$ at any value $b$. Thus honest coverage
\eqref{eqn:standardHonest} implies Definition~\ref{def:honest}.
The proof of Theorem~\ref{thm:lowerbound} invokes this size control at the tested perturbation
$A_n$, coverage of $A_n$ under the confounding distribution whose parameter is $A_n$, which any honest
interval provides, so the bound applies to every honest interval.

The primitive driving the bound is that the class can confound a shift in
$\beta$ with a change in the nuisance. 
This is the information-theoretic dual of the
worst-case bias $\bar B$, where $\bar B$ measures the largest shift in $\hat\beta$ the
class can induce.
The condition below measures the largest shift in $\beta$ the class can \emph{hide in the likelihood}. 
Define the total variation distance, 
\begin{align*}
    TV(P, Q) : = \sup_{A\in\mathcal{F}}|P(A) - Q(A)|.
\end{align*}

\begin{definition}[Confounded pairs]\label{def:confpair}
Fix $n$. Two admissible index points $(\beta_0,\theta_0)$ and $(\beta_1,\theta_1)$, with
$\theta_0,\theta_1\in\mathcal{C}_n$ and $\beta_0\neq\beta_1$, are \emph{confounded at level}
$t\in[0,1)$ if $TV\big(P^{(n)}_{\beta_0,\theta_0},P^{(n)}_{\beta_1,\theta_1}\big)\le t$, and
\emph{exactly confounded} if $t=0$, in which case the two pairs generate the same distribution.
\end{definition}

Note that
$(\beta,\theta)\mapsto P^{(n)}_{\beta,\theta}$ need not be one-to-one. Confounding is precisely the
failure of injectivity in the target coordinate, namely two pairs, different targets, one distribution. The identifying content of the class is which
target gaps admit such pairs, and the next assumption posits them at every gap up to the
reach.

\begin{assumption}[Least-favourable confounding]\label{ass:confounding}
    There is a function $\tau:[0,1]\to[0,1-\alpha)$, continuous and nondecreasing with
    $\tau(0)=0$, such that for every $n$ and every $\Delta\in[0,2\bar B]$ there exist
    $\theta_0,\theta_1\in\mathcal{C}_n$ with $P_{\beta_0,\theta_0}^{(n)},
    P_{\beta_0+\Delta,\theta_1}^{(n)}\in\mathcal{P}_n$ and
    \begin{align*}
        TV\left(\,P_{\beta_0,\theta_0}^{(n)},P_{\beta_0+\Delta,\theta_1}^{(n)}\right)\ \leq\ \tau\!\left(\Delta/2\bar B\right) + o(1),
    \end{align*}
    where the $o(1)$ term is uniform in $\Delta$.
\end{assumption}

In the language of Definition~\ref{def:confpair}, at every target gap $\Delta\le2\bar B$
the class contains a pair confounded at level $\tau(\Delta/2\bar B)+o(1)$.
Assumption~\ref{ass:confounding} thus says that a coefficient shift of up to \emph{twice} the
worst-case bias can be absorbed by the class, with detectability governed by the fraction
$\Delta/2\bar B$. The factor of two is Figure~\ref{fig:deadzone}(b). To get the factor of two, ignore sampling
noise and suppose the data reveal only where $\hat\beta$ is centred. Under the true distribution that
centre is $\beta_0+\hat B$ with $|\hat B|\le\bar B$, and honesty requires an interval to keep
every value within $\bar B$ of it, since each of them could be the truth under some admissible
bias. In the least-favourable configuration $\hat B=+\bar B$, the truth $\beta_0$ sits at the
lower end of these values, and the values that cannot be rejected reach up to $\beta_0+2\bar B$. The value $\beta_0+\Delta$ lies within $\bar B$ of the centre only if
\begin{align*}
    \Delta\ \le\ \hat B+\bar B\ \le\ 2\bar B,
\end{align*}
with $\Delta=2\bar B$ attained at $\hat B=+\bar B$.
Beyond $2\bar B$ the value lies more than $\bar B$ from the centre under every admissible bias,
and the separation between the distributions is forced positive. 
The bias is untestable and the
coverage promise is uniform over both of its directions, so where $\beta_0$ sits relative to the
centre is unknown. At $\hat B=-\bar B$ the truth sits at the upper end and the reach runs down to
$\beta_0-2\bar B$. Tracing every position in between pushes the reach out in both directions to
the dead zone $|\Delta|\le2\bar B$, whose half-width is the same about any hypothesised value.
A vanishing shift is
asymptotically undetectable ($\tau(0)=0$), while a shift equal to the full reach remains
at most borderline detectable, $\tau(1)<1-\alpha$, so that even a full-reach shift leaves
the power of any honest test bounded away from one (in the exactly confounded settings of
Lemma~\ref{lem:exactconf}, $\tau\equiv0$ and shifts up to the full reach are completely
undetectable, the total variation being zero at every $n$). 

The restriction to $[0,1-\alpha)$ is
costless in the cases considered herein. Section~\ref{sect:primitives} shows that when the confounding
is measured through the modulus, the Neyman--Pearson analogue of this bound has
$\tau(u)=\Phi(s(u)-z_{1-\alpha})-\alpha<1-\alpha$ because $\Phi<1$. The calibration in terms of the ratio
$\Delta/2\bar B$, rather than $\Delta n^{\epsilon}$, is exactly the statement that the
worst-case bias corresponds to $O(1)$ units of statistical separation, i.e. the
modulus-of-continuity relation $\bar B \asymp \omega(n^{-1/2})$ of
\citet{donoho1994statistical}. Here $\omega(\varepsilon)$, the modulus of continuity, is the
largest amount by which the target can differ between two members of the class whose data
distributions are at statistical distance $\varepsilon$. 
The function $\tau$ bundles the shape of the modulus with
the two-point affinity.

The results below come in two tiers. The general
tier leaves $\tau$ free, with bound $\alpha+\tau(\cdot)$, and covers approximate confounding.
The exact tier sets $\tau\equiv0$, e.g. in plausibly-exogenous IV and honest DiD
(Lemma~\ref{lem:exactconf}) the confounded distributions coincide at every $n$, every $\alpha+\tau$
bound collapses to $\alpha$, and the dead zone is flat. Regression discontinuity designs are exactly confounded through the committed estimator (the scope of
Assumption~\ref{ass:LE} and Proposition~\ref{prop:sharp}) but only approximately confounded
over the full class, where the frontier is the smoothly rising one of
Section~\ref{sect:primitives}. The weak-factor model claims only $\tau(0)=0$. 

\begin{theorem}[Honest power]\label{thm:lowerbound}
    Make Assumptions~\ref{ass:BiasBound}--\ref{ass:confounding}, use
    $n^{\delta}\bar B\to\bar b\in(0,\infty)$, and use Definition~\ref{def:LocalAlt}. Let
    $\mathcal{B}_n$ be any honest interval \eqref{eqn:standardHonest}, with data generated at
    $\beta_0$, and let its dual test reject the local alternative
    $A_n=\beta_0+\xi\,n^{-\rho}$ when $A_n\notin\mathcal{B}_n$. Fix $\xi\neq 0$.
    \begin{enumerate}[(i)]
    \item If $\delta<\rho$, there is a nuisance $\theta_0\in\mathcal{C}_n$ such that, under the
    true distribution $P_{\beta_0,\theta_0}^{(n)}$,
    \begin{align*}
        \limsup_{n\to\infty}\ \Pr\left\{A_n\notin\mathcal{B}_n\right\}\ \leq\ \alpha.
    \end{align*}
    No honest interval excludes $n^{-\rho}$-perturbations of the truth with more than its size.
    \item If $\delta=\rho$ and $|\xi|< 2\bar b$, the same construction yields $\theta_0$ with
    \begin{align*}
        \limsup_{n\to\infty}\ \Pr\left\{A_n\notin\mathcal{B}_n\right\}\ \leq\ \alpha + \tau\!\left(|\xi|/2\bar b\right)\ <\ 1.
    \end{align*}
    The power to exclude a perturbation of the truth is bounded away from one throughout the
    window $|\xi|<2\bar b$. It exceeds $\alpha$ only where $\tau(|\xi|/2\bar b)>0$, and under
    the exact confounding of Lemma~\ref{lem:exactconf} the ceiling is $\alpha$ itself
    (Corollary~\ref{cor:flatdz}).
    \end{enumerate}
\end{theorem}

\begin{proof}
    Fix $\xi$ and set $\Delta_n=|\xi|\,n^{-\rho}$. The argument is symmetric in the sign
    of $\xi$. For any two probability measures $P_0,P_1$ and any test $\phi$,
    \begin{align*}
        \mathrm{E}_{P_0}\phi - \mathrm{E}_{P_1}\phi
        = \int \phi\,(dP_0 - dP_1)
        \leq \int_{dP_0>dP_1}(dP_0-dP_1)
        = TV(P_0, P_1),
    \end{align*}
    since $0\leq\phi\leq1$. In both cases below $\Delta_n\leq2\bar B$ for $n$ large, so
    Assumption~\ref{ass:confounding} supplies $\theta_0,\theta_1$ with
    $P_{0,n}:=P_{\beta_0,\theta_0}^{(n)}$ the true distribution ($\beta=\beta_0$) and
    $P_{1,n}:=P_{\beta_0+\Delta_n,\theta_1}^{(n)}\in\mathcal{P}_n$ a confounding distribution whose own
    parameter is $A_n=\beta_0+\Delta_n$. Apply the displayed inequality to the dual test
    $\phi=\mathbbm{1}\{A_n\notin\mathcal{B}_n\}$ with $P_0=P_{0,n}$ and $P_1=P_{1,n}$. Honesty of
    $\mathcal{B}_n$ at $A_n$, the parameter of $P_{1,n}$, controls the size of this test,
    $\mathrm{E}_{P_{1,n}}\phi\leq\alpha+o(1)$, so
    \begin{align*}
        \mathrm{E}_{P_{0,n}}\{A_n\notin\mathcal{B}_n\}
        \ \leq\ \alpha + TV\left(P_{0,n}, P_{1,n}\right) + o(1)
        \ \leq\ \alpha + \tau\!\left(\Delta_n/2\bar B\right) + o(1).
    \end{align*}
    Under $P_{0,n}$ the true value is $\beta_0$. For (i), $\delta<\rho$ gives
    $\Delta_n/2\bar B=\Theta(n^{\delta-\rho})\to0$, and continuity with $\tau(0)=0$ yields the bound
    $\alpha$. For (ii), $\delta=\rho$ gives $\Delta_n/2\bar B\to|\xi|/2\bar b$, and continuity yields
    $\alpha+\tau(|\xi|/2\bar b)$, which is below one because $\tau(|\xi|/2\bar b)\le\tau(1)<1-\alpha$.
    Taking $\theta_0$ as above completes the proof.
    \hfill
\end{proof}

\begin{lemma}[Exact confounding]\label{lem:exactconf}
Assumption~\ref{ass:confounding} holds with $\tau\equiv0$ on $[0,1]$, the total variation
distance being exactly zero at every $n$, in the two settings below.
\begin{enumerate}[(i)]
\item Plausibly-exogenous IV (Example~\ref{ex:PlausExog}), with $\bar B=\bar\gamma/\pi$.
For $\Delta\in[0,2\bar\gamma/\pi]$ the distributions $(\beta_0,\ \gamma=\bar\gamma,\ \varepsilon)$ and
$(\beta_0+\Delta,\ \gamma=\bar\gamma-\Delta\pi,\ \varepsilon'=\varepsilon-\Delta v)$ induce the
same joint distribution of $(y,x,z)$, and $|\bar\gamma-\Delta\pi|\le\bar\gamma$ exactly when
$\Delta\le2\bar\gamma/\pi$. The redefined error is admissible because the class leaves the
dependence between the structural and first-stage errors unrestricted.
\item Honest DiD over $\Delta^{SD}(M)$ (Example~\ref{ex:DiD}), with $\bar B=M$, exactly the
worst-case bias of the extrapolation estimator (over the class the raw post coefficient's bias
is unbounded, as linear trends have zero second differences, which is why the honest procedure
extrapolates). For $\Delta\in[0,2M]$ the trend paths $(0,\dots,0,\ \delta_1=M)$ at target
$\beta_0$ and $(0,\dots,0,\ \delta_1=M-\Delta)$ at target $\beta_0+\Delta$ give identical
event-study means with zero pre-period coefficients, and both lie in $\Delta^{SD}(M)$ exactly
when $\Delta\le2M$. The post-period violation consistent with flat pre-periods ranges over
$[-M,M]$, so the identified set for $\beta$ has diameter $2M$.
\end{enumerate}
\end{lemma}
\begin{proof}
In each case the two distributions induce the same distribution of the data, so $TV=0$, while the
targets differ by $\Delta$.
\hfill
\end{proof}

\begin{corollary}[Flat dead zone]\label{cor:flatdz}
In the settings of Lemma~\ref{lem:exactconf}, with no limit experiment or regularity beyond the
class itself, every honest interval satisfies $\limsup_n\Pr\{A_n\notin\mathcal{B}_n\}\le\alpha$
for every sequence of alternatives with $|\Delta_n|\le2\bar B$ eventually. At the matched rate
the minimax dead zone thus has half-width exactly $2\bar b$, and at fixed bounds it eventually
contains every local alternative.
\end{corollary}
\begin{proof}
The bound is the two-point inequality of Theorem~\ref{thm:lowerbound} with zero total
variation. The half-width is exact because for $\rho<\epsilon$ the bias-aware interval excludes
every $|\xi|>2\bar b$ wpa1 (Theorem~\ref{thm:nondeg}(i)), and for $\rho=\epsilon$ the attained
frontier of Proposition~\ref{prop:sharp} exceeds $\alpha$ at every $|\xi|>2\bar b$.
\hfill
\end{proof}

At the matched rate $\delta=\rho=\epsilon$ the bound of Theorem~\ref{thm:lowerbound}(ii) is
sharp under the following condition, localising $\beta=\beta_0+h\,n^{-\epsilon}$ with the
local coefficient $h\in\mathbb{R}$ of Definition~\ref{def:LocalAlt} as a free coordinate and
the perturbation of interest at $h=\xi$.

\begin{assumption}[Gaussian least-favourable limit]\label{ass:LE}
Let $\delta = \epsilon$.
There is a nuisance path $g\mapsto\theta_n(g)\in\mathcal{C}_n$ and a scale
$\sigma\in(0,\infty)$ such that the localised experiments
$\{P^{(n)}_{\beta_0+h\,n^{-\epsilon},\theta_n(g)}\}$ converge, in the sense of
\citet{lecam1972limits} \citep[see also][Chs.~9 and~15]{vandervaart1998asymptotic}, to the bounded-normal-mean experiment of observing
$X\sim N(h+m,\sigma^2)$, the nuisance contributing a free location shift
$m\in[-\bar b,\bar b]$ with $\bar b=\lim_n n^{\delta}\bar B$, and $\hat\beta$ is
asymptotically sufficient, $n^{\epsilon}(\hat\beta-\beta_0)\xrightarrow{d}X$.
\end{assumption}

Assumption~\ref{ass:LE} is the exact-confounding form of Assumption~\ref{ass:confounding}. In
the settings of Lemma~\ref{lem:exactconf} the coefficient vector is Gaussian and shifts in the
target and in the nuisance are exact location shifts, so the localised experiments equal the
bounded-normal-mean experiment at every $n$ and the estimator carries all the usable
information. The assumption therefore holds in these settings whenever the bound is scaled to the matched rate
(for instance $M_n\asymp n^{-\epsilon}$ in Example~\ref{ex:DiD}).

In regression discontinuity the assumption fails for the full dataset but holds for the
committed estimator, a local-linear estimate with a given kernel and bandwidth. The
distribution of the committed estimator depends on the regression function only through where the
estimator is centred, the target plus a bias within $\pm\bar B$, so every value within $\bar B$
of that centre could be the truth and the confounding, approximate in the data, is exact
through the committed estimator.
Proposition~\ref{prop:sharp} therefore binds every test built from it. Over all tests the frontier rises smoothly (Proposition~\ref{prop:confprim}) but lifts off
with zero slope at the truth (Corollary~\ref{cor:rdcheap}), so the confounding remains close to
exact where the dead zone binds.

In the weak-factor model of Example~\ref{ex:WeakFactors}, Assumption~\ref{ass:confounding}
formalises that a factor at the boundary of detection is not consistently detectable by any test
\citep{onatskimoreirahallin2013}, its principal-components estimate being inconsistent
\citep{onatski2012}. A coefficient shift aligned with its loading is then absorbed into the
nuisance, the pair realising target gaps up to the reach $2\bar B$. This confounding is
approximate and only $\tau(0)=0$ is claimed, so the operative bound is the $\alpha+\tau$ of
Theorem~\ref{thm:lowerbound}.

\begin{proposition}[Sharp honest power]\label{prop:sharp}
Under Assumption~\ref{ass:LE}, with data at truth
$\beta_0$, every honest interval excludes the perturbation $\beta_0+\xi n^{-\epsilon}$ with
worst-case power (over the nuisance) no larger than
$\pi_\alpha(\xi):=\Phi\!\big((|\xi|-2\bar b)_+/\sigma-z_{1-\alpha}\big)$. In particular
$\pi_\alpha(\xi)=\alpha$ for every $|\xi|\le 2\bar b$, so no honest interval excludes a
perturbation within a half-width $2\bar b$ of the truth.
\end{proposition}
\begin{proof}
By \citet{lecam1972limits} it suffices to argue in the limit $X\sim N(\nu,\sigma^2)$,
$\nu=h+m$, $|m|\le\bar b$ (the truth $\beta_0$ sits at $h=0$ by the centring of the localisation, a change of origin only). Take $\xi\ge0$. If $\xi\le 2\bar b$, the nuisance $m'=-\bar b$
gives alternative mean $\nu_1=\xi-\bar b\in[-\bar b,\bar b]$, an admissible null mean, so
honesty forces power $\le\alpha$. If $\xi>2\bar b$, honesty gives the dual test size at most $\alpha$ at the target-$\xi$ distribution
with the least favourable mean $\nu_1=\xi-\bar b$, and the Neyman--Pearson lemma bounds its
rejection probability at the truth distribution with mean $\nu_0=\bar b$ by
$\Phi\big((\xi-2\bar b)/\sigma-z_{1-\alpha}\big)$, the two means being $(\xi-2\bar b)/\sigma$
standard deviations apart. The bound is attained by the dual test rejecting $\xi$ when
$X<\xi-\bar b-z_{1-\alpha}\sigma$. It has size at most $\alpha$ at every admissible
target-$\xi$ mean, and its rejection probability at $\nu_0=\bar b$ is exactly the frontier. The two-sided case differs by a
negligible tail \citep[Le~Cam's two-point method;][]{tsybakov2009}.
\hfill
\end{proof}

The standardised separation $(|\xi|-2\bar b)_+/\sigma$ in $\pi_\alpha$ is the confounding
separation of Definition~\ref{def:confsep} (Section~\ref{sect:primitives}) evaluated in this limit experiment, zero throughout
the dead zone because the confounding is exact. Under approximate confounding it is strictly
increasing, and the frontier of Proposition~\ref{prop:confprim} rises smoothly instead.

In the same limit the bias-aware interval reproduces the dead zone exactly, with worst-case
power $\Phi\big((|\xi|-2\bar b)/\sigma-z_{1-\alpha/2}\big)$ outside it up to a negligible tail,
and misses $\pi_\alpha$ only by the critical-value gap $z_{1-\alpha/2}-z_{1-\alpha}$, which
the directional critical value $\mathrm{cv}_\alpha$ of \citet{armstrong2018optimal} removes
(there $\mathrm{cv}_\alpha(t)-t\to z_{1-\alpha}$). This is the asymptotic counterpart of the
finite-sample finding of \citet{armstrong2022robust} that the bias-aware critical value cannot
be reduced by more than roughly a factor of two without sacrificing coverage. By the
length-power duality of \citet{pratt1961length}, Proposition~\ref{prop:sharp} is the
local-power counterpart of the length non-adaptivity of \citet{low1997nonparametric} and
\citet{cailow2004adaptation}, and it rests on the same two-point modulus construction of
\citet{donoho1994statistical} as their length bound.

\begin{definition}[Power dead zone]\label{def:deadzone}
With the data at the truth $\beta_0$ and the bias untestable, the \emph{dead zone} of an honest
interval $\mathcal{B}_n$ is the set of local perturbations of the truth whose exclusion cannot be
guaranteed beyond size. With $A_n=\beta_0+\xi\,n^{-\rho}$,
\begin{align*}
\mathcal{D}\;=\;\Big\{\,\xi:\ \limsup_{n\to\infty}\ \inf_{\theta\in\mathcal C_n}\
\Pr_{P^{(n)}_{\beta_0,\theta}}\{A_n\notin\mathcal{B}_n\}\le\alpha\,\Big\}.
\end{align*}
Inside $\mathcal{D}$ there is an admissible bias under which the perturbation is excluded no more
often than the size of the dual test, so its exclusion cannot be promised to carry evidence.
\end{definition}

At bias $c\bar b$ the window of values the bias-aware interval with $\rho<\epsilon$ cannot
exclude is $[(c-1)\bar b,(c+1)\bar b]$ (Theorem~\ref{thm:nondeg}(i)), the values within $\bar b$
of the rescaled centre $c\bar b$ of $\hat\beta$, which the shaded regions of the simulations
realise. At $c=1$ the truth sits at the lower end of the window, which reaches $2\bar b$. Since
the bias is untestable and may take either sign, only the union of these windows can be
promised, and this union is the dead zone (Figure~\ref{fig:deadzone}(b)),
\begin{align*}
\bigcup_{c\in[-1,1]}\big[(c-1)\bar b,\ (c+1)\bar b\big]\;=\;[-2\bar b,\ 2\bar b].
\end{align*}
Figure~\ref{fig:unionzone} realises this union by simulation, the dead zone being an infimum
over biases.

\begin{figure}[!ht]
\centering
\includegraphics[width=0.75\textwidth]{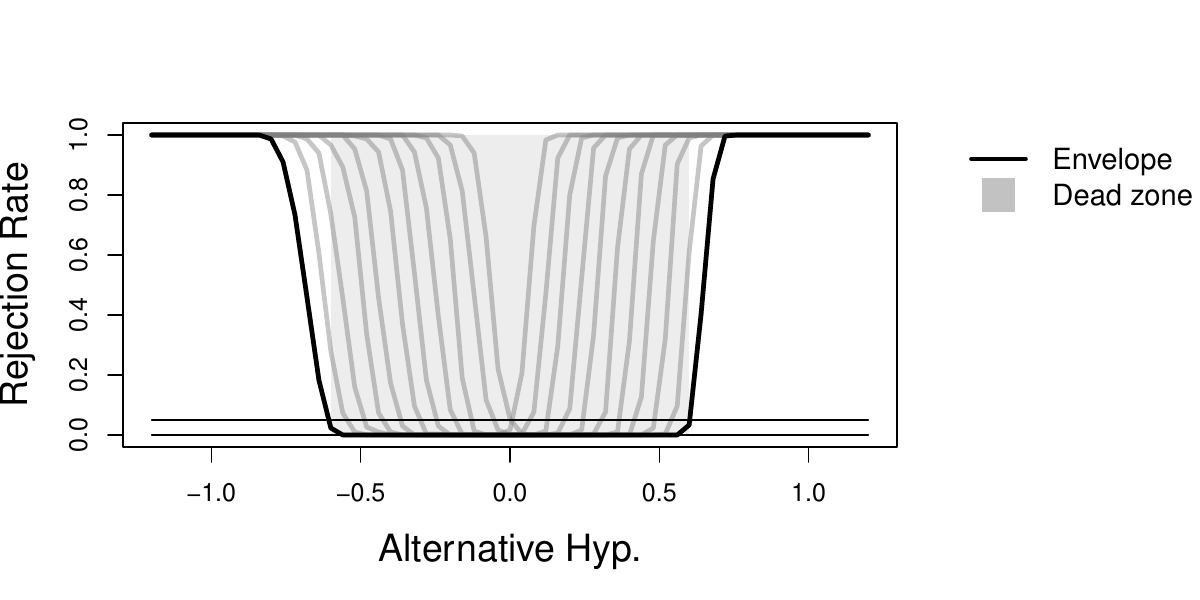}
\caption{The union dead zone in the plausibly-exogenous IV design of
Section~\ref{sect:simulations}, with $\bar\gamma=0.3$, $\pi=1$ (so $\bar b=0.3$) and $n=1000$.
Grey lines are rejection rates under $\gamma=c\bar\gamma$, $c\in[-1,1]$. The black line,
their pointwise minimum over $c$ and the empirical counterpart of Definition~\ref{def:deadzone},
is flat on exactly $\pm2\bar b$ and rises outside it. Shading marks where the minimum is at most $\alpha$, namely
$[-0.6,0.6]=\pm2\bar b$.}
\label{fig:unionzone}
\end{figure}

\begin{corollary}[Rate-optimality of the bias-aware interval]\label{cor:rateoptimal}
    Under the assumptions of Theorem~\ref{thm:lowerbound}, the bias-aware interval
    $\mathcal{B}_n$, whose dual test excludes $A_n$ when $A_n\notin\mathcal{B}_n$, attains the
    transition rate of the bound. By Theorems~\ref{thm:LocalPower}--\ref{thm:nondeg} its power to
    exclude perturbations of the truth is asymptotically zero for $\delta<\rho$ and non-degenerate
    for $\delta=\rho$. For $\delta>\rho$ with $\rho<\epsilon$, the rescaled half-width vanishes
    (immediately, $n^{\rho}\bar B=\Theta(n^{\rho-\delta})\to0$) while $|\hat B|\le\bar B$ pins the rescaled centre at the truth,
    so Lemma~\ref{lem:containment} gives exclusion of every fixed $\xi\neq0$ wpa1. Hence honest
    local power transitions at $\rho=\delta$,
    and no honest interval improves on this rate.
\end{corollary}

The mechanism is more general than honest inference. Let a partially identified $\beta$ have identified set $\mathrm{ID}(P)$, the values of $\beta$
consistent with the data distribution $P$ over the maintained class.

\begin{proposition}[Partial identification forces local-power degeneracy]\label{prop:partialID}
Let $\mathrm{ID}(P)$ be a compact interval of positive length, let the data be generated at
the truth $\beta_0$ (distribution $P$), and let $\mathcal{B}_n$ be any interval that is \emph{locally
valid at $P$}, $\Pr_P\{b_n\in\mathcal{B}_n\}\ge1-\alpha+o(1)$ for every sequence
$b_n\in\mathrm{ID}(P)$ with $b_n\to\beta_0$. For perturbations
$A_n=\beta_0+\xi\,n^{-\rho}$ with $\rho>0$:
\begin{enumerate}[(i)]
\item if $\beta_0\in\operatorname{int}\mathrm{ID}(P)$, then for every fixed $\xi$,
$A_n\in\mathrm{ID}(P)$ for $n$ large and $\limsup_n\Pr_P\{A_n\notin\mathcal{B}_n\}\le\alpha$, so the interval cannot exclude perturbations, and there is no local power;
\item if $\beta_0$ is an endpoint of $\mathrm{ID}(P)$, the same holds for every $\xi$
pointing into $\mathrm{ID}(P)$, so there is at best one-sided local power.
\end{enumerate}
\end{proposition}
\begin{proof}
If $\beta_0$ and $A_n$ both lie in $\mathrm{ID}(P)$ they are consistent with the same distribution
$P$, so under the truth $\beta_0$ the value $A_n$ is an equally admissible parameter. Local
validity at $P$ therefore forces the interval to cover it, $\Pr_P\{A_n\in\mathcal{B}_n\}\ge1-\alpha+o(1)$,
whence $\Pr_P\{A_n\notin\mathcal{B}_n\}\le\alpha+o(1)$. In (i), $A_n\to\beta_0\in
\operatorname{int}\mathrm{ID}(P)$ gives $A_n\in\mathrm{ID}(P)$ for $n$ large, for every
$\xi$. In (ii), the same holds for $\xi$ pointing inward, while perturbations leaving the
set are realised by distinct distributions and may be excludable.
\hfill
\end{proof}

Honesty enters Proposition~\ref{prop:partialID} only as a sufficient condition, since every
honest interval is locally valid at every admissible distribution. Under exact confounding a fixed bound
is partial identification, the identified set being the values within $\bar B$ of where
$\hat\beta$ is centred. With the truth at its lower endpoint, a case of part~(ii), the set
reaches $2\bar B$ above the truth. Since the bias is untestable the truth may sit anywhere in
the set, which extends the reach to both sides, so its diameter $2\bar B$ is the dead-zone
half-width in the units of the estimand. The regimes $0<\delta\le\rho$ are partial
identification at the local scale, since the rescaled half-width $n^{\rho}\bar B$ does not
vanish and in-window perturbations are realisations of the same distribution
\citep[cf.][]{imbens2004confidence,stoye2009more}. Under approximate confounding $\beta$ may
remain point identified. The bound in part~(i) is $\alpha$ rather than zero, which is sharp,
because honesty caps the exclusion power in excess of size, not the exclusion probability itself.

Against a class wide enough to contain the bias, no honest procedure recovers the local power
that conservatism forgoes, and regaining it requires relaxing strict honesty, for instance
through an undersized test that retains power
\citep[cf.][]{low1997nonparametric,rambachan2023more}.

\subsection{Primitive conditions and the modulus}\label{sect:primitives}

The matched-rate normal approximation of Theorem~\ref{thm:nondeg}(ii) and
Proposition~\ref{prop:sharp} holds when $\hat\beta$ is asymptotically linear,
$\hat\beta-\beta-\hat B=\sum_i w_{ni}u_i+o_p(n^{-\epsilon})$, with $\{w_{ni}u_i\}$ satisfying
the Lindeberg condition and $n^{2\epsilon}\sum_i w_{ni}^2\,\mathrm{E}u_i^2\to\sigma^2$. For
$\epsilon=1/2$ this is the usual Lyapunov central limit theorem, and for $\epsilon<1/2$ the
standard Lindeberg condition for kernel estimators with effective sample size $n^{2\epsilon}$,
which holds for the estimators of the simulation designs in Section~\ref{sect:simulations}.
It also needs a consistent standard error, $n^{\epsilon}se(\hat\beta)\xrightarrow{p}\sigma$, and
a stable bias, $n^{\delta}\hat B\xrightarrow{p}c\bar b$ and $n^{\delta}\bar B\to\bar b$. These
conditions certify the Gaussian limit for tests built from the estimator. The every-honest-test
scope of Proposition~\ref{prop:sharp} also needs the asymptotic-sufficiency clause of
Assumption~\ref{ass:LE}.

The model implies $\tau$ in Assumption~\ref{ass:confounding} through the
confounding separation.

\begin{definition}[Confounding separation]\label{def:confsep}
For $u\in[0,1]$ let the admissible pairs at gap $2u\bar B$ be the pairs
$\big(P^{(n)}_{\beta_0,\theta_0},\,P^{(n)}_{\beta_0+2u\bar B,\theta_1}\big)$ with
$\theta_0,\theta_1\in\mathcal C_n$ (Definition~\ref{def:confpair}). Define
\begin{equation*}
s_n(u)\ :=\ \inf_{\theta_0,\theta_1\in\mathcal C_n}\
2\,\Phi^{-1}\!\Big(\tfrac12\Big[1+TV\big(P^{(n)}_{\beta_0,\theta_0},\,P^{(n)}_{\beta_0+2u\bar B,\theta_1}\big)\Big]\Big).
\end{equation*}
This is the smallest total variation across pairs, expressed as a mean gap in standard-error
units of the Gaussian location pair with that same total variation. In Gaussian-shift
settings this is $\sqrt n\,\inf\norm{f_1-f_0}/\sigma$, the rescaled inverse of the
\citet{donoho1994statistical} modulus of continuity. The \emph{confounding separation} is
$s(u)=\lim_n s_n(u)$, assumed continuous and nondecreasing with $s(0)=0$. When the argument is the local coefficient, write $s(\xi):=s(|\xi|/2\bar b)$.
\end{definition}

Each target value is carried by its slice, the class of distributions
$\{P^{(n)}_{b,\theta}:\theta\in\mathcal C_n\}$. Both hypotheses are therefore composite in the
nuisance, and achievable worst-case power is set by the closest cross-pair of distributions. The infimum
is not an extra assumption but the operational meaning of ``least favourable''. With a singleton
class $\mathcal C_n=\{0\}$ the frontier below is the Neyman--Pearson two-point power function
$\Phi(s_n(u)-z_{1-\alpha})$.

\begin{proposition}[Primitive condition for the bound]\label{prop:confprim}
Suppose the confounding separation $s(u)$ of Definition~\ref{def:confsep} exists and the
localised experiments carry the Gaussian-shift structure, so that each admissible pair at gap
$2u\bar B$ is, in the limit, a Gaussian location pair whose standardised mean gap equals its
separation under Definition~\ref{def:confsep}. This holds in the white-noise settings of
\citet{donoho1994statistical}, and trivially under the exact confounding of
Lemma~\ref{lem:exactconf}, where $s$ is identically zero. Then every
asymptotically honest level-$\alpha$ test has power at most $\alpha+\tau(u)$ at the confounded
alternative, with
\begin{align*}
\tau(u)=\Phi\big(s(u)-z_{1-\alpha}\big)-\alpha,
\end{align*}
which has $\tau(0)=0$ and, since $\Phi<1$, lies in $[0,1-\alpha)$. 
Moreover, over convex classes and linear functionals the bound is \emph{attained}. The
likelihood-ratio test between the closest pair of the two convex slices is exactly honest over
the tested slice, by the separating-hyperplane property of the closest pair, and its worst-case
power is exactly $\Phi(s(u)-z_{1-\alpha})$, so
$\Phi\big(s(\xi)-z_{1-\alpha}\big)$ is the exact honest local-power frontier, attained
pointwise.
\end{proposition}
\begin{proof}
Honesty caps size at every admissible distribution carrying the alternative, so for any
admissible pair $(P_0,P_1)$ at gap $2u\bar B$, by the Gaussian-shift hypothesis a Gaussian
location pair with standardised mean gap $s$, the Neyman--Pearson lemma gives, where
subscript $0$ marks the distribution under the true parameter values,
\begin{align*}
\mathrm E_{P_0}\phi\ \le\ \sup\{\mathrm E_{P_0}\psi:\ \mathrm E_{P_1}\psi\le\alpha\}
\ =\ \Phi\big(s-z_{1-\alpha}\big),
\end{align*}
increasing in $s$. Taking the infimum over admissible pairs bounds power by
$\Phi(s_n(u)-z_{1-\alpha})=\alpha+\tau_n(u)$, replacing the cruder
$\mathrm E_{P_0}\phi\le\alpha+TV$ step in Theorem~\ref{thm:lowerbound}.

Write $\mathcal F_0$ and $\mathcal F_1$ for the two slices, the sets of
admissible mean objects carrying targets $\beta_0$ and $\beta_0+2u\bar B$ (convex, by
convexity of the class and linearity of the target). 
Let $(f_0^*,f_1^*)$ attain the closest pair between the convex slices and
$\delta=f_1^*-f_0^*$. Minimality of the pair gives the variational inequalities
$\langle f_1-f_1^*,\delta\rangle\ge0$ on $\mathcal F_1$ and
$\langle f_0-f_0^*,\delta\rangle\le0$ on $\mathcal F_0$. Write $dY$ for the Gaussian-shift observation, the data as mean object plus white noise at
scale $\sigma/\sqrt n$, so that the $\delta$-weighted average of the data satisfies
$\langle\delta,dY\rangle\sim N(\langle\delta,f\rangle,\ \sigma^2\norm{\delta}^2/n)$ under
mean object $f$. The test rejecting when
$\langle\delta,dY\rangle<c$, with $c$ set by size $\alpha$ at $f_1^*$, therefore has size at
most $\alpha$ on all of $\mathcal F_1$, and its worst-case power over $\mathcal F_0$ is
attained at $f_0^*$ and equals $\Phi(\sqrt n\,\norm{\delta}/\sigma-z_{1-\alpha})
=\Phi(s_n(u)-z_{1-\alpha})$. Both limits follow by continuity of $\Phi$.
\hfill
\end{proof}

Without the Gaussian-shift structure, existence of $s(u)$ alone bounds power only through the
total-variation step of Theorem~\ref{thm:lowerbound}, $\alpha+TV=\alpha+(2\Phi(s(u)/2)-1)$, which a
two-point pair with that total variation attains.


The four settings, Examples~\ref{ex:WeakFactors}-\ref{ex:PlausExog} and regression discontinuity, differ in how $s$ behaves at the origin. Exact confounding makes the dead zone
flat, while the strictly increasing $s$ of regression discontinuity leaves no flat dead zone over
the full class. 
The plausibly exogenous IV and honest DiD settings both adhere to the exact confounding of Lemma~\ref{lem:exactconf}, and exhibit the flat dead zone behaviour inside the $2\bar B$ gap, i.e. $s(u) = 0$ for all $u\in[0,1]$. 
For the weak factor setting only $s(0) = 0$ is claimed as per the above discussion around the work in \cite{onatski2012,onatskimoreirahallin2013}. 
As seen below in Lemma~\ref{lem:pricecurve} and Corollary~\ref{cor:rdcheap}, power loss is cheap in the RD setting, with $s(u) = u^{5/4}$, resulting from the inverse H\"older modulus in \cite{donoholow1992}. 


In regression discontinuity let $g=f_1-f_{-1}$ for two candidate regression functions with
$|f''|\le M$ (Figure~\ref{fig:confsep}), with $g(x)=0$ for $x<0$ and $g(0)=\Delta$ at the
cut-off without loss. Then $|g''|\le2M$, and the closest pair minimises $\norm{g}_2$, the distributions'
mean separation being $\approx\sqrt n\norm{g}_2/\sigma$. If $g(x)\neq0$ on $[0,h)$ and $g(x)=0$
for $x\ge h$ then $g(h)=g'(h)=0$, $g'$ being Lipschitz, and a Taylor expansion gives
\begin{align*}
\Delta=g(0)=\tfrac{h^2}{2}\,g''(\tilde h)\le Mh^2,
\end{align*}
so $h\ge\sqrt{\Delta/M}$, and the minimising $g$ spends maximal curvature throughout,
$g''\equiv2M$ on $(0,h)$, which is the parabola $g(x)=\Delta(1-x/h)^2$ with $h=\sqrt{\Delta/M}$. Hence, with
$\Delta=2u\bar B$ and the RD bias bound $\bar B\asymp M^{1/5}(\sigma/\sqrt n)^{4/5}$,
\begin{align*}
\norm{g}_2\asymp\sqrt{h}\,\Delta\asymp M^{-1/4}\Delta^{5/4},
\qquad
s_n(u)\asymp\frac{\sqrt n}{\sigma}\,M^{-1/4}(2u\bar B)^{5/4}\asymp u^{5/4}.
\end{align*}

\begin{figure}[!ht]
\centering
\includegraphics[width=0.95\textwidth]{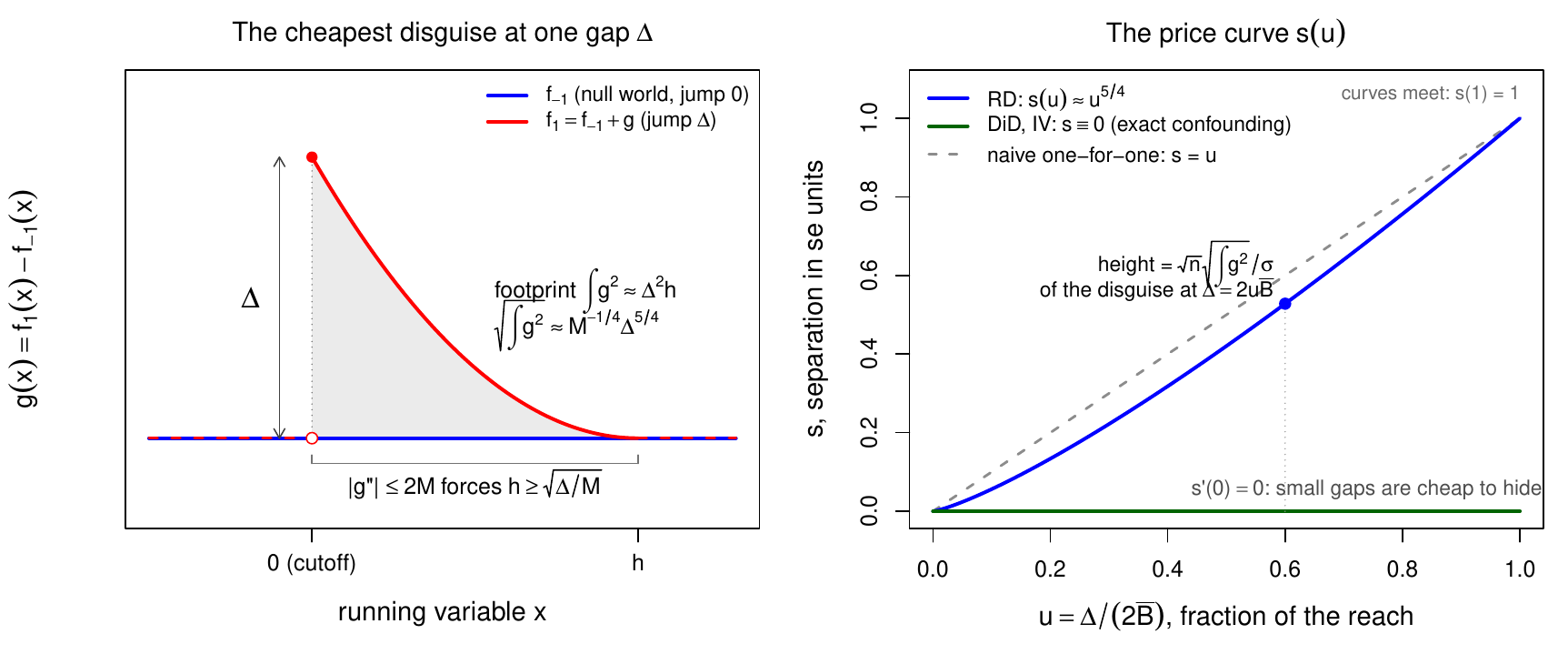}
\caption{The confounding separation, primal and dual. Left, at one gap $\Delta$ the
least-favourable regression functions merge as fast as the curvature bound allows, and the best
test detects the shaded footprint with power $\Phi(\sqrt n\,\norm{g}_2/\sigma-z_{1-\alpha})$.
Right, the price curve $s(u)$ over gaps $\Delta=2u\bar B$, zero under exact confounding (DiD,
plausibly exogenous IV) and $u^{5/4}$ in regression discontinuity, with the naive $s=u$
dashed.}
\label{fig:confsep}
\end{figure}

Lemma~\ref{lem:pricecurve} (Figure~\ref{fig:funnel}) records what the modulus geometry of
\citet{donoho1994statistical} and the noise-level calibration of the bias bound jointly imply for
the power frontier. This implication does not appear to have been stated in this form.

\begin{lemma}[Power law]\label{lem:pricecurve}
Let the experiment be the Gaussian-shift (white-noise) model over a convex class with a
linear target functional, and let the modulus $\omega$ be
regularly varying at zero with index $r\in(0,1]$, as holds with $r=4/5$ for the H\"older
class $\{|f''|\le M\}$ \citep{donoholow1992}. Calibrate the bound at the matched
rate, so that $s(1)=\lim_n\sqrt n\,\omega^{-1}(2\bar B)/\sigma\in(0,\infty)$. Then the
confounding separation of Definition~\ref{def:confsep} satisfies
\begin{align*}
s(u)=s(1)\,u^{1/r},\qquad u\in(0,1],
\end{align*}
and the honest local-power frontier of Proposition~\ref{prop:confprim} is
$\Phi\big(s(1)\,u^{1/r}-z_{1-\alpha}\big)$.
\end{lemma}
\begin{proof}
Under the Gaussian-shift structure the least-separated admissible pair whose targets differ
by $\Delta$ lies at $L^2$ distance $\omega^{-1}(\Delta)$. The modulus is by definition the
largest target gap available at a given $L^2$ distance, so its inverse is the smallest $L^2$
distance at which a given gap is available. Hence
$s_n(u)=\sqrt n\,\omega^{-1}(2u\bar B)/\sigma$ by Definition~\ref{def:confsep}. Since
$\omega$ is regularly varying at zero with index $r$, its inverse is regularly varying with
index $1/r$, so $\omega^{-1}(2u\bar B)/\omega^{-1}(2\bar B)\to u^{1/r}$ as $\bar B\to0$ (the matched-rate
calibration enforces this, since $s(1)$ finite requires $\omega^{-1}(2\bar B)\asymp n^{-1/2}\to0$,
hence $\bar B=\Theta(\omega(n^{-1/2}))\to0$),
uniformly over $u$ in compact subsets of $(0,1]$ by the uniform convergence theorem for
regularly varying functions. Multiplying by $s_n(1)\to s(1)$ gives the result.
\hfill
\end{proof}

\begin{figure}[!ht]
\centering
\includegraphics[width=0.98\textwidth]{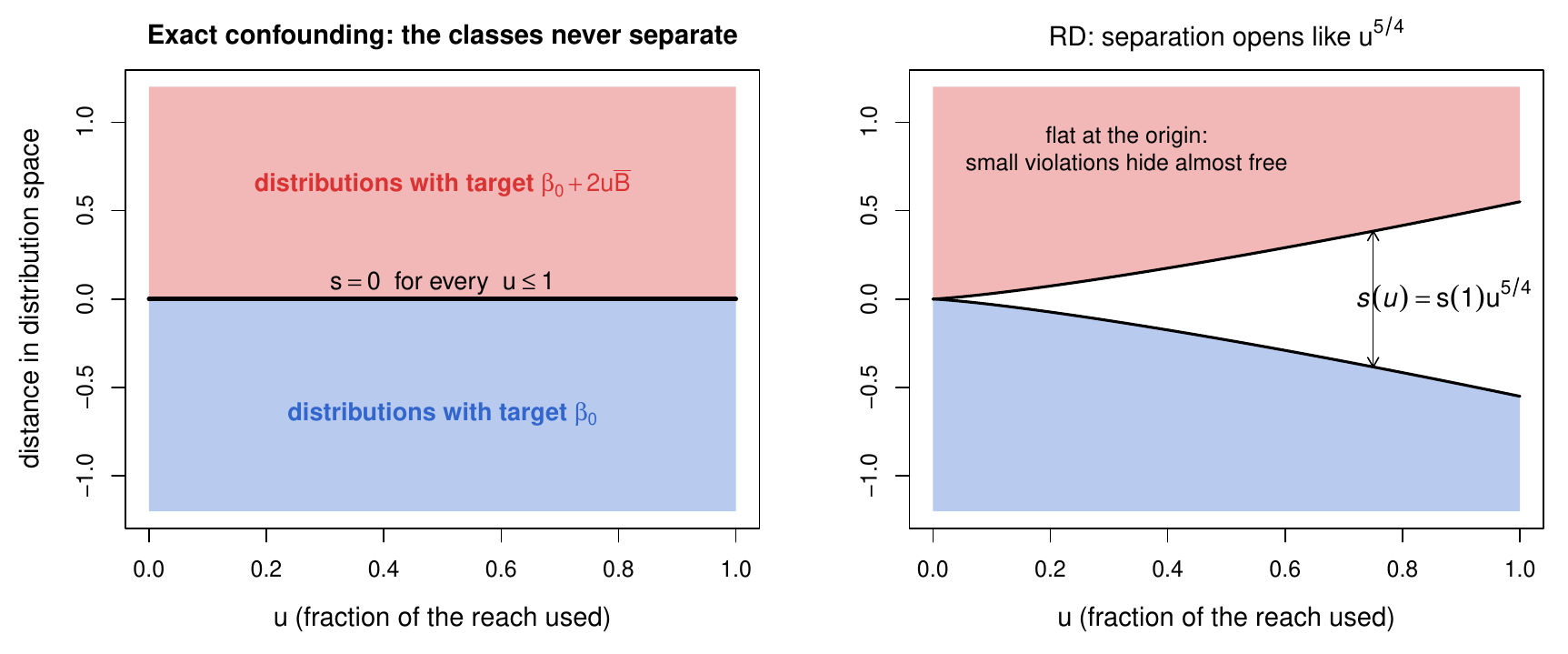}
\caption{Price curve of Lemma~\ref{lem:pricecurve} over $u$. Left, exact confounding, $s\equiv0$.
Right, regression discontinuity, where $s$ widens as $u^{5/4}$, flat at the origin, so
the frontier at $\alpha$ has zero slope (Corollary~\ref{cor:rdcheap}).}
\label{fig:funnel}
\end{figure}

\begin{corollary}[Honest power in RD]\label{cor:rdcheap}
For the H\"older class $\{|f''|\le M\}$ the frontier over all honest tests is
$\Phi\big(s(1)\,u^{5/4}-z_{1-\alpha}\big)$. The frontier over the full class is nowhere flat, and the ceiling exceeds $\alpha$ at every $u>0$, some test tuned to that alternative beating
$\alpha$, while tests built from the committed estimator remain at exactly $\alpha$ inside the
window (Proposition~\ref{prop:sharp}). Also, no honest test lifts off with positive slope.
The ceiling itself is nearly flat near the truth, so the power forgone by the bias-aware
class inside its bounded window is at most $\Phi(s(1)u^{5/4}-z_{1-\alpha})-\alpha$, which is small
throughout the window when $s(1)$ is moderate. Finally, outside the window the bias-aware test
attains the limit-experiment frontier of Proposition~\ref{prop:sharp} up to the fixed
critical-value gap $z_{1-\alpha/2}-z_{1-\alpha}$, which the directional critical value of
\citet{armstrong2018optimal} removes. Honesty in regression discontinuity is therefore cheap
in power.
\end{corollary}
\begin{proof}
For $|f''|\le M$ the modulus has $r=4/5$ \citep{donoholow1992}, so
Lemma~\ref{lem:pricecurve} gives $s(u)=s(1)u^{5/4}$ and Proposition~\ref{prop:confprim}
turns this into the frontier $\Phi(s(1)u^{5/4}-z_{1-\alpha})$, valid for every honest test attained pointwise. The first reading is from $s(1)u^{5/4}>0$ for $u>0$ and $\Phi$
is strictly increasing. 
The second since the same expression is the ceiling on any test's
power inside the window, while every test built from the estimator has power exactly
$\alpha$ there (Proposition~\ref{prop:sharp}), so the forgone power is their difference. 
The
third is Proposition~\ref{prop:sharp} applied to the estimator's limit experiment, where the
bias-aware interval's power is $\Phi\big((|\xi|-2\bar b)_+/\sigma-z_{1-\alpha/2}\big)$,
below the one-sided frontier by exactly the critical-value gap.
\hfill
\end{proof}

The difference $g$ is why the estimator's scope matters in regression discontinuity and not in
difference-in-differences. The reported estimate $\hat\beta=\sum_iw_iy_i$ is nearly blind to $g$, whereas the
Neyman--Pearson weighting $w_i=g(x_i)$ attains the full mean gap $\sqrt n\,\norm{g}_2/\sigma$,
the Cauchy--Schwarz maximum. 
Such a statistic estimates the treatment effect jump poorly, hence likely never used in practice. 
However, the frontier is an
impossibility claim, so $s(u)$ allows the statistic to be chosen freely, and the power forgone is zero under exact confounding and at most
$\Phi(s(1)u^{5/4}-z_{1-\alpha})-\alpha$ in regression discontinuity.

\section{Empirical applications}\label{sect:applications}

The paper's practical recommendation is to report the power ``dead zone'' alongside
bias-aware intervals. 
This is the set of effect sizes the bound forgoes the ability to detect. 
This section demonstrates the recommendation on two easily reproducible published studies.\footnote{Replication data is in the \texttt{RDHonest} \citep{kolesar2026rdhonest} and
\texttt{HonestDiD} \citep{rambachanroth2026honestdid} packages. The
regression-discontinuity data are those of \citet{lee2008randomized}, as distributed in
\citet{angristpischke2009mhedata}.} 
The applications sit in two of the paper's regimes, namely bounded and cheap at the matched rate for regression discontinuity, and
decisive and non-vanishing when the bound dominates for DiD.

In both applications the relevant null is no effect, $\beta_0=0$, and the dead zone is the set of
true effects whose detection cannot be guaranteed.
That is, the zone symmetric about $\beta_0=0$ that an honest test rejects no more often than its size. 
The dead-zone half-width $2\bar B$ is a property of
the bound, the distance from the truth an honest test cannot be guaranteed to reject when the untestable bias sits at the bound. It has the same width around \emph{any} hypothesised value, so it is drawn at the economically natural $\beta_0=0$. 

\subsection{Regression discontinuity: incumbency advantage}\label{subsec:app-rd}

Consider the \citet{lee2008randomized} regression-discontinuity estimate of the U.S. House
incumbency advantage, using the bias-aware procedure of \citet{armstrong2018optimal,
armstrong2020simple} with the data-driven smoothness constant of its default rule of thumb. The
point estimate is a $5.85$ percentage-point incumbency advantage. The conventional interval,
which drops the smoothing bias, is $(3.17,8.53)$. The honest interval, which inflates the
critical value from $z_{1-\alpha/2}=1.96$ to the bias-aware $\mathrm{cv}_\alpha=2.31$, is
$(2.69,9.01)$, only $18\%$ wider. This is the matched-rate case $\delta=\rho=\epsilon=2/5$ of
Theorem~\ref{thm:nondeg}. The worst-case bias and the standard error are of the same order
($\bar B/se\approx0.65$), so conservatism costs only a bounded amount. The worst-case bias is
$\bar B=0.89$, so the dead zone of the reported procedure has half-width
$2\bar B=1.78$ percentage points.
Here the estimate exceeds the dead-zone
half-width by a factor of more than three (Figure~\ref{fig:app_rd}), so the finding is robust to
the conservatism, and disclosing the dead zone confirms rather than overturns it.

\begin{figure}[!ht]
  \centering
  \includegraphics[width=0.8\textwidth]{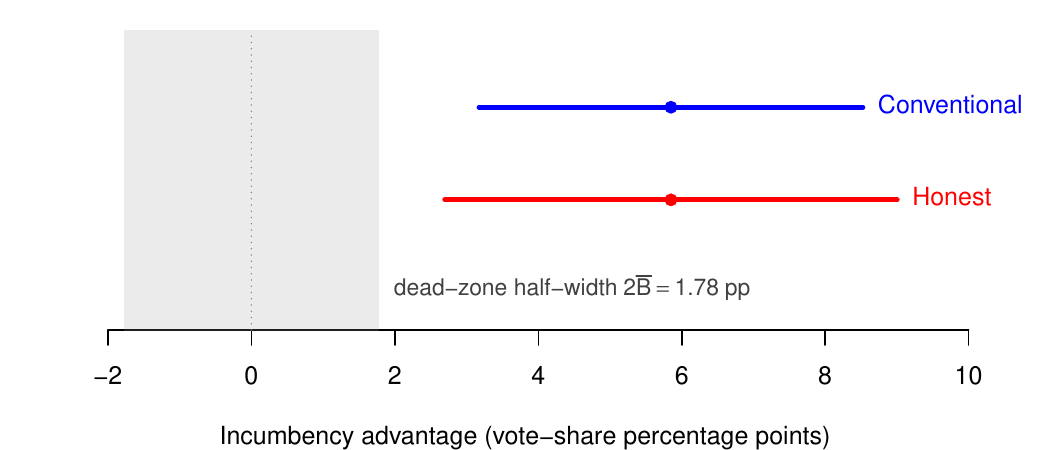}
  \caption{Incumbency advantage \citep{lee2008randomized}. Conventional (blue) and 
  bias-aware (red) confidence intervals, with the bias-aware dead zone
  half-width $2\bar B=1.78$ percentage points, shaded about the null $\beta_0=0$ (the band has
  the same width about any null). The estimate ($5.85$ pp) lies far outside the dead zone,
  detectable despite the conservatism.}
  \label{fig:app_rd}
\end{figure}

\subsection{Difference-in-differences: VAT pass-through to profits}\label{subsec:app-did}

Consider the \citet{benzarti2019who} event study of a large French VAT cut, asking whether the
reduction was captured as firm profits, with the honest difference-in-differences procedure of
\citet{rambachan2023more}. Imposing exact parallel trends, the first post-reform year ($2009$)
shows a sharply significant profit effect of $0.196$ with conventional interval $(0.159,0.233)$
($t\approx10$). 
Under the relative-magnitudes restriction
$\Delta^{RM}(\bar M)$, which allows the post-period violation of parallel trends to be at most a
multiple $\bar M$ of the largest consecutive pre-period violation, the honest interval widens with $\bar M$
and does not shrink with the sample. 

The lower CI limit crosses zero at 
$\bar M\approx1.75$ (Figure~\ref{fig:app_did}).
\citet{rambachan2023more} analyse this application with the same machinery and report a
closely related sensitivity analysis. 
This reporting, however, ignores the dead-zone which extends to twice the bias bound. 
An effect the size of the point estimate enters the
dead zone of no-effect at $\bar M\approx1.23$ (the grey fan in Figure~\ref{fig:app_did}), where $2\bar B(\bar M)$ reaches the
estimate ($\bar B(\bar M)=\bar M\times0.0795$, where $0.0795$ is the largest consecutive
pre-period deviation and $0.196$ the estimate). Past that bound, detection of an effect of that size can no longer be
guaranteed. 

The two thresholds differ because the realised interval is built around the estimate actually observed, whereas guaranteed detection must allow the bias to pull the centre of $\hat\beta$ towards zero by up to $\bar B(\bar M)$, from where the interval reaches down a further $\bar B(\bar M)$ (Section~\ref{sect:minimaxpowerLossHonest}). Guaranteed detectability of the
conventionally significant effect is thus lost once one concedes post-period trend deviations
roughly $1.2$ times the largest one seen before treatment. The realised rejection survives
further, to $\bar M\approx1.75$.
Disclosure of the dead zone in this application makes explicit that the ability to reject a zero effect rests on the bound the analyst is prepared to defend, i.e. on the prescribed tolerance for pre-trend violations.

\begin{figure}[!ht]
  \centering
  \includegraphics[width=0.8\textwidth]{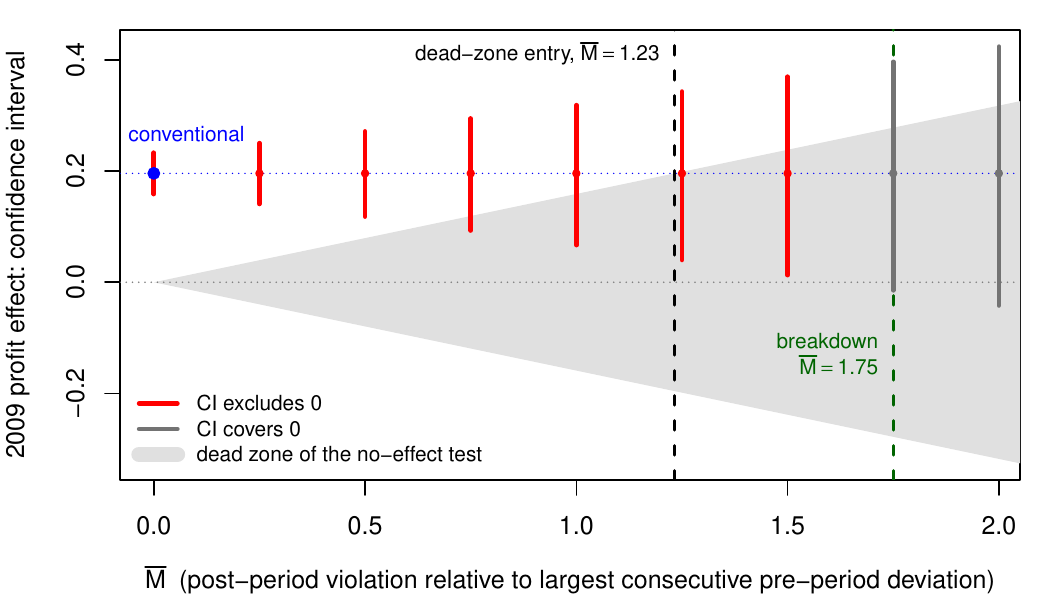}
  \caption{VAT pass-through to profits \citep{benzarti2019who}. The honest interval for the
  $2009$ profit effect as a function of the relative-magnitudes bound $\bar M$. Red bars exclude
  zero and grey bars cover zero. The grey fan is the dead zone of the no-effect test,
  $\pm2\bar B(\bar M)=\pm2\times0.0795\,\bar M$, which widens with the bound and crosses the
  estimate (blue dotted) at $\bar M\approx1.23$ (black dashed). Past that bound, detection of an
  effect the size of the estimate can no longer be guaranteed, although the realised interval
  rejects zero until $\bar M\approx1.75$ (green dashed). The conventional interval
  is the $\bar M=0$ point. The dead zone is set by $\bar M$, not by the sample size.}
  \label{fig:app_did}
\end{figure}

In regression discontinuity the dead zone is a
bounded, modest object and the headline survives it. In difference-in-differences the dead zone
is governed entirely by an untestable bound, and can be wide enough that even a very large $t$-statistic is not, by itself, decisive. In both cases the dead-zone half-width is a non-rejection region
that the interval's endpoints alone do not convey, and should be reported along other estimator diagnostics.

\section{Simulations}\label{sect:simulations}

Each of the five designs compares an honest procedure with a conventional one that ignores the
worst-case bias. In each replication the data are generated at the truth $\beta_0$, and the dual
test rejects the local alternative $A_n=\beta_0+\xi\,n^{-\rho}$ when $A_n\notin\mathcal B_n$.
Averaging over replications traces the local-power curve in $\xi$, with $\alpha=0.05$. Notation
inside the subsections follows each design's literature ($\tau$ for the regression-discontinuity
jump, $\rho$ for the endogeneity correlation, $s$ and $w$ for trend parameters), not the symbols
of the theory sections. The factor, difference-in-differences and plausibly-exogenous designs are
bound-dominated (Theorems~\ref{thm:LocalPower}--\ref{thm:LocalCI2}), regression discontinuity is
the matched-rate case, and in the linear-functional design the honest interval converges more
slowly than the conventional one.\footnote{The factor designs use the
\texttt{PanelIFE} package \citep{armstrong2024panelife}.}

\subsection{Interactive fixed effects with no, weak and strong factors}
\label{sect:simFactor}

With $n=NT$, the interactive fixed-effects panel of \cite{armstrong2022robust} is generated as
\begin{equation*}
  y_{it}=x_{it}\beta+c_n\,\lambda_i f_t+\varepsilon_{it},
  \qquad
  x_{it}=\lambda_i f_t+\eta_{it},
  \qquad i\le N,\ t\le T,
\end{equation*}
with $\lambda_i,f_t,\varepsilon_{it},\eta_{it}\sim\mathrm{i.i.d.}\,N(0,1)$, the product
$\lambda_i f_t$ normalised to unit variance, $\beta=1$ and $T=\lfloor N^{3/5}\rfloor$. The factor
is absent from the outcome at $c_n=0$, weak at $c_n=n^{-1/3}$ and strong at $c_n=1$ (the rows of
Figure~\ref{fig:Factors}), and local alternatives are $A_n=\beta\pm\xi/\sqrt{n}$.

The iterated least-squares \emph{factor} estimator of \cite{Bai2009}, with interval
$\hat\beta\pm z_{1-\alpha/2}\,\mathrm{se}(\hat\beta)$, and the \emph{bias-aware} interval
$\{\hat\beta_{ba}\pm[\bar B+z_{1-\alpha/2}\,\mathrm{se}(\hat\beta_{ba})]\}$ of
\cite{armstrong2022robust} both allow $R=2$ factors, at least one more than the truth contains. The
\emph{pre-test} applies the eigenvalue ratio of \cite{ahn2013eigenvalue} to the least-squares
residuals, with a mock zeroth eigenvalue permitting $\hat R_{strong}=0$. Its count
$\hat R_{strong}$ is consistent for the number of strong factors and does not count weak ones.
The pre-test reports least squares if $\hat R_{strong}=0$ and the factor estimator with
$\hat R_{strong}$ factors otherwise, each with its conventional interval, and it sets the
bias-aware declaration $R_{weak}=\max\{R-\hat R_{strong},0\}$.

Under the weak factor the estimator over-rejects the truth, and the pre-test, with
detection rate essentially zero, collapses to least squares, which shares the factor estimator's
biased rejection curve because the estimated factor directions are uninformative. 
With no factor the pre-test selects $\hat R_{strong}=0$ and is nominally sized, and
under the strong factor it coincides with the factor estimator for $N\ge200$. Under the strong factor the
bias-aware interval is conservative, insuring a second, hypothetical weak factor, and with no
factor it pays this dead-zone premium with nothing to insure. No test can tell apart the no-factor and weak-factor
designs \citep{onatskimoreirahallin2013} and the
premium is unavoidable.

\begin{figure}[!ht]
    \centering
    \includegraphics[width=0.32\linewidth, height = 5cm]{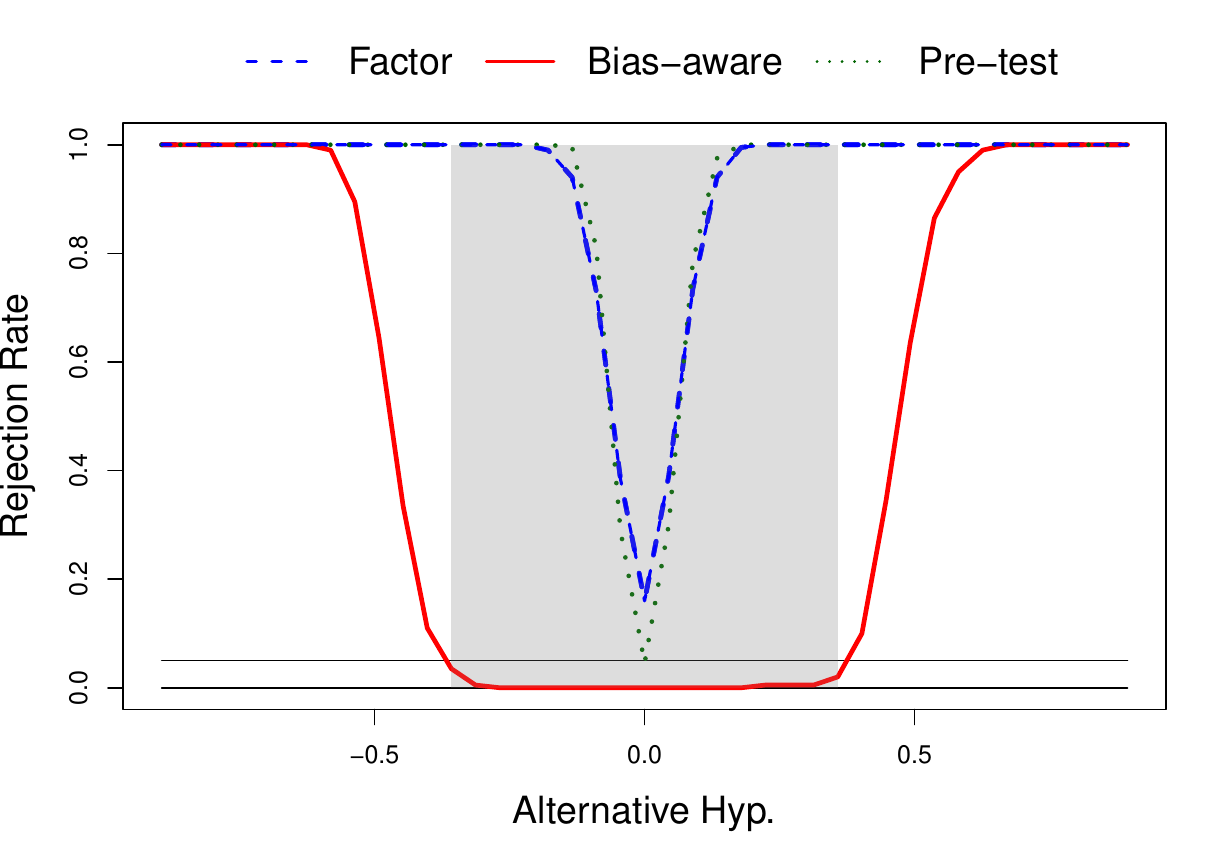}
    \includegraphics[width=0.32\linewidth, height = 5cm]{NoPrateNonRejectionRegionsN=200T=24.pdf}
    \includegraphics[width=0.32\linewidth, height = 5cm]{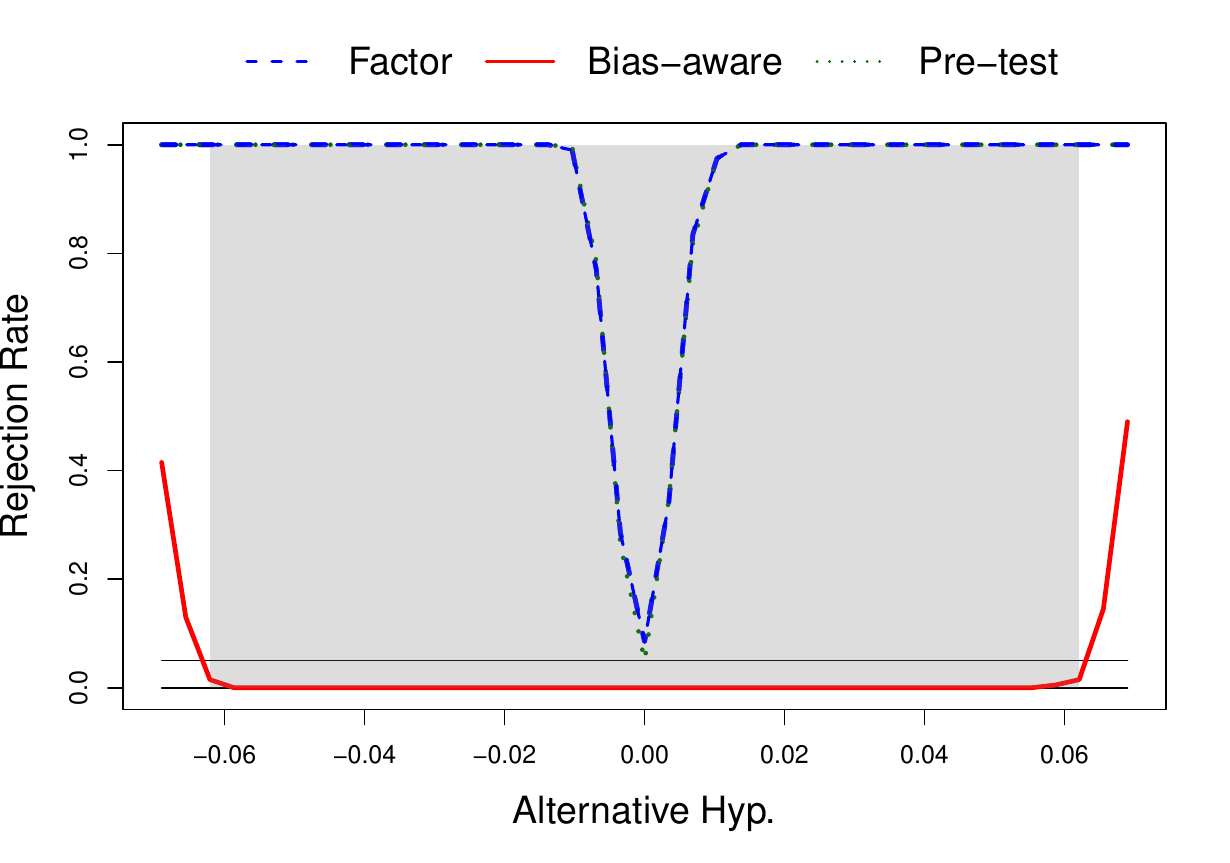}
    \centering
    \includegraphics[width=0.32\linewidth, height = 5cm]{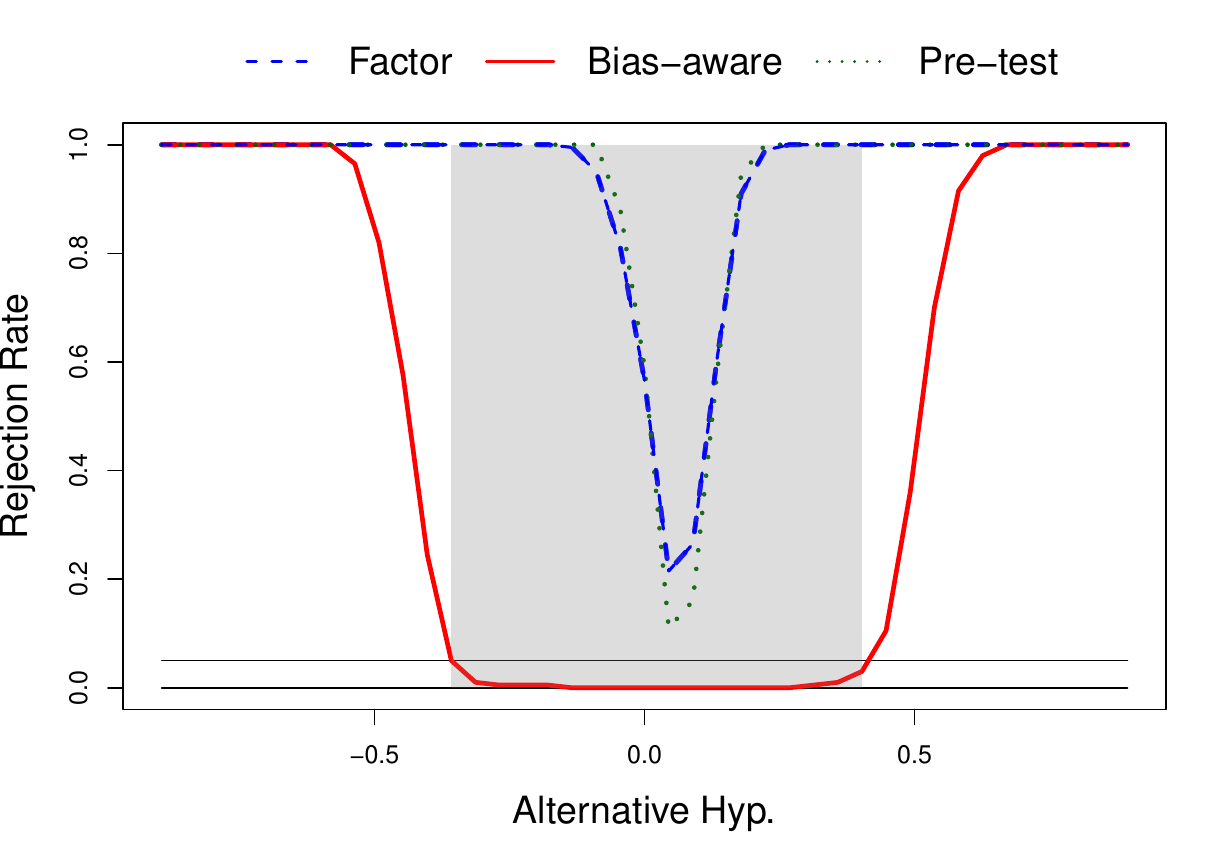}
    \includegraphics[width=0.32\linewidth, height = 5cm]{WeakPrateNonRejectionRegionsN=200T=24.pdf}
    \includegraphics[width=0.32\linewidth, height = 5cm]{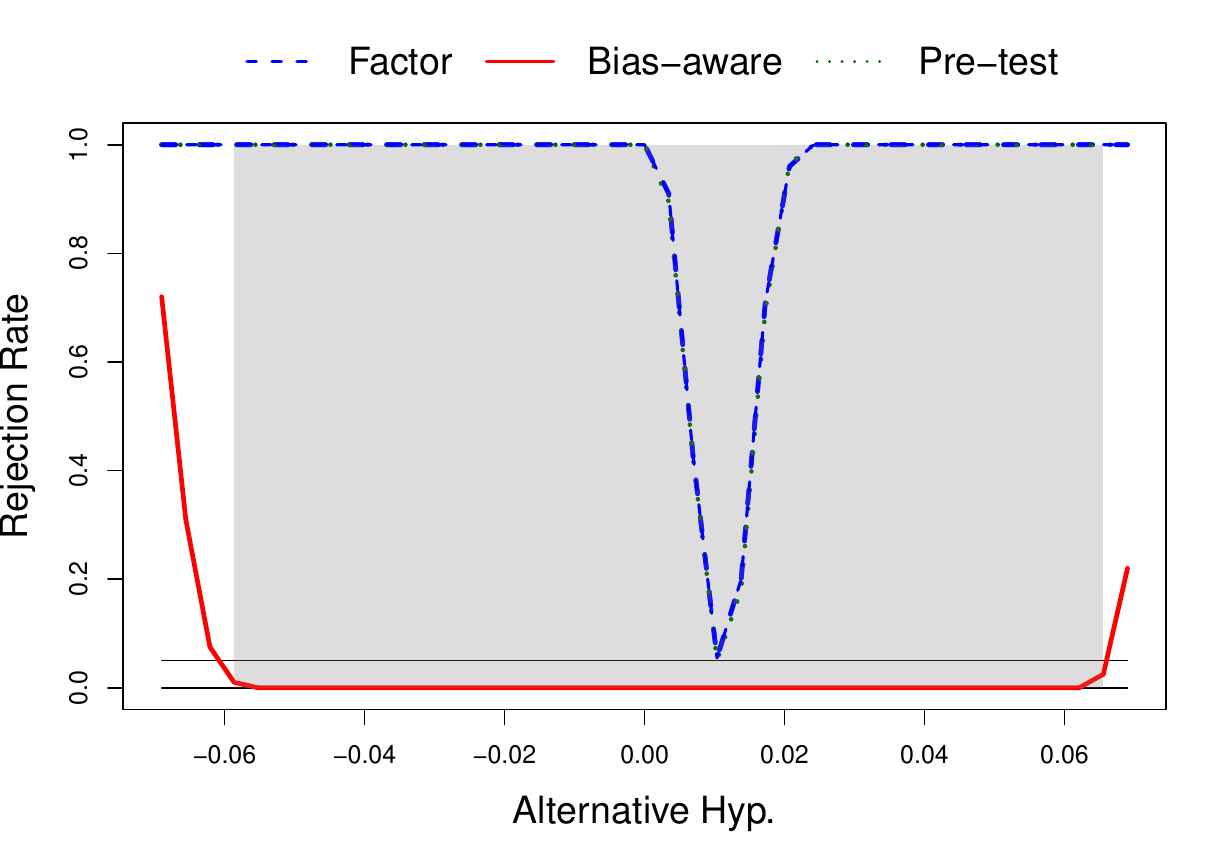}
    \centering
    \includegraphics[width=0.32\linewidth, height = 5cm]{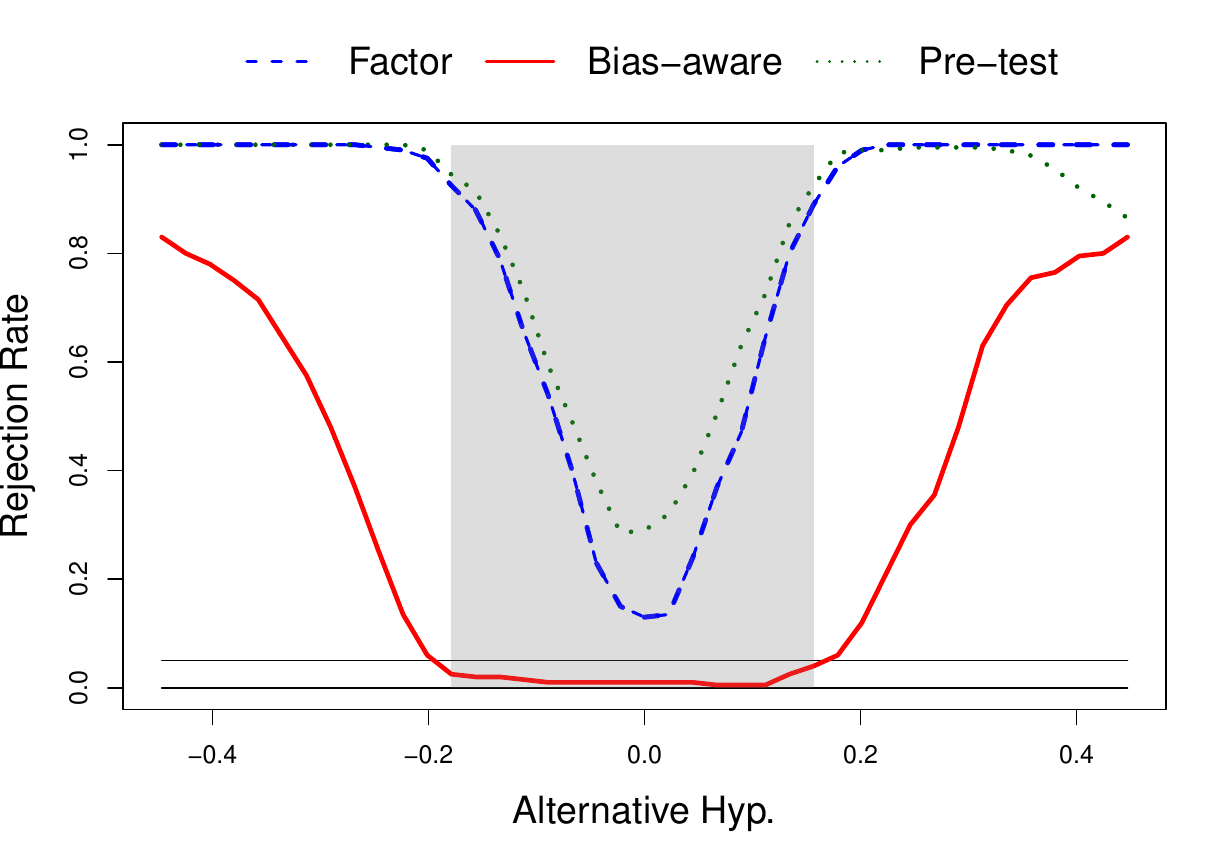}
    \includegraphics[width=0.32\linewidth, height = 5cm]{StrongPrateNonRejectionRegionsN=200T=24.pdf}
    \includegraphics[width=0.32\linewidth, height = 5cm]{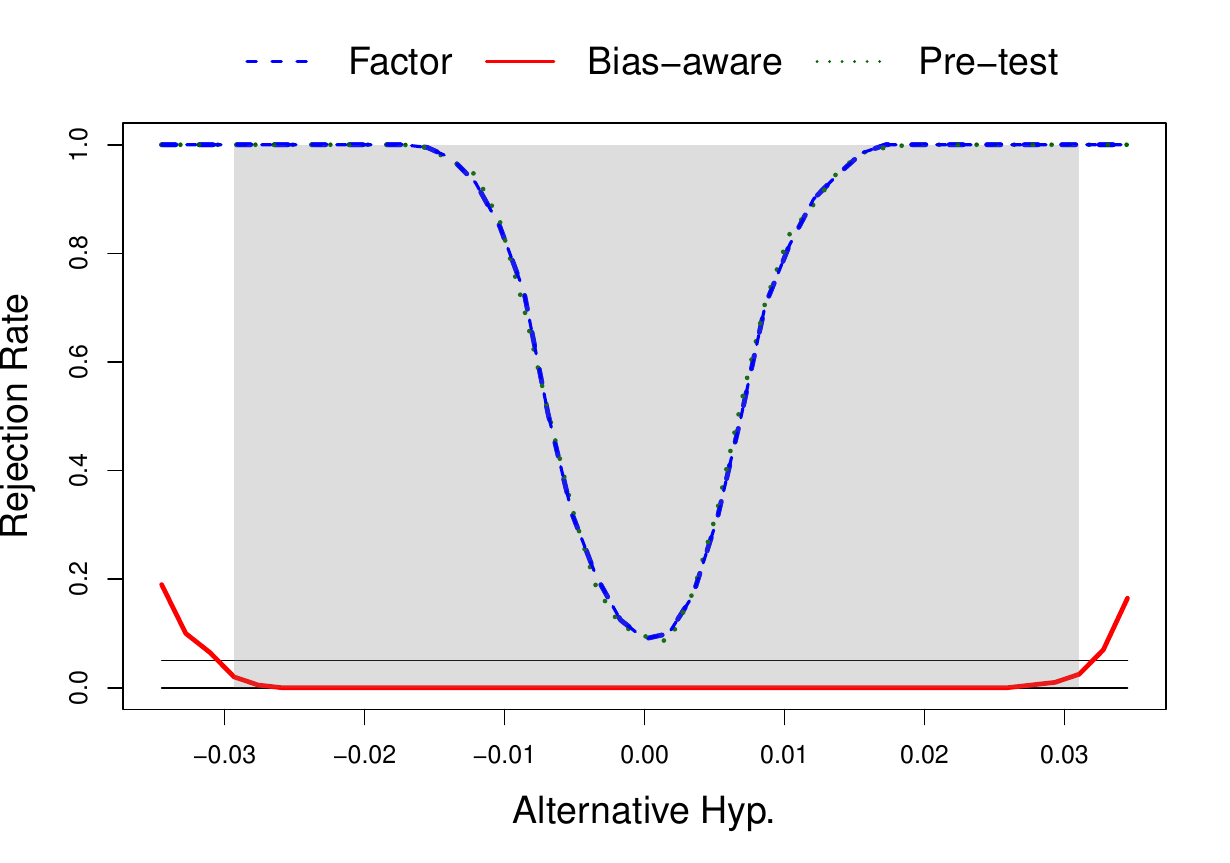}
    \caption{Interactive fixed-effects rejection rates for $n^{-1/2}$ local alternatives.
    Rows are no, weak and strong factor, columns $(N,T)\in\{(50,10),(200,24),(1200,70)\}$.
    Lines show the factor estimator (blue dashed), bias-aware interval (red solid) and pre-test
    (green dotted). Grey shading marks the realised dead zone $\{$power $\le\alpha\}$, with two factors
    declared weak in the no and weak rows and one in the strong row.
    }
    \label{fig:Factors}
\end{figure}

\subsection{Honest difference-in-differences}
\label{subsec:did}

In the honest difference-in-differences setting of \cite{rambachan2023more}, the coefficient
vector $\hat\beta\in\mathbb{R}^4$ of an event study with three pre-periods and one post-period is
drawn as
\begin{equation*}
  \hat\beta\sim N\!\left(\mu,\ \Sigma_0/N\right),
  \qquad
  \mu=\big(-3s,\ -2s,\ -s,\ \beta_{\mathrm{post}}+s+w\big),
  \qquad \Sigma_0=I_4,
\end{equation*}
with true post-period effect and target $\theta=\beta_{\mathrm{post}}=0$. The mean encodes a
linear differential trend $\delta_t=s\cdot t$ over event time $t\in\{-3,-2,-1,1\}$ ($t=0$
normalised out) plus a second-difference kink $w$ at $t=1$, with $s,w\ge0$. Local alternatives
are $A_n=\beta_{\mathrm{post}}\pm\xi/\sqrt N$, and Figure~\ref{fig:did} reports a strong
pre-trend ($s=0.2$, $w=0$), a kinked violation ($s=0$, $w=0.2$) and no pre-trend ($s=w=0$).

The \emph{event-study} interval $\hat\theta\pm z_{1-\alpha/2}\,\mathrm{se}(\hat\theta)$ imposes
parallel trends ($\delta=0$). The \emph{honest} interval is the fixed-length interval over
$\Delta^{SD}(M)$ with $M=0.2$, centred at the extrapolation-corrected estimator and widened by the
worst-case post-period bias. The \emph{trend-adjusted} interval is the construction over
$\Delta^{SD}(0)$, extrapolating the linear pre-trend with no allowance for curvature, as is
common in applied work \citep{dobkin2018economic,bhuller2013broadband}.

Under the strong pre-trend the event-study interval centres on $\beta_{\mathrm{post}}+s$ and
loses coverage as $N$ grows, while the honest interval extrapolates the trend away and so is
recentred at the truth rather than widened about the event-study centre. Under an exactly linear
trend the trend-adjusted interval is correctly centred and nominally sized and keeps full local
power, bought by asserting the untestable $w=0$. The kinked violation is the least-favourable
path of Lemma~\ref{lem:exactconf}(ii), on which the event-study and trend-adjusted intervals are
both miscentred by $M$ and reject the truth with probability approaching one. The honest
interval, at its designed worst case, spends exactly its size $\alpha$ at the truth for $N\ge500$ ($0.040$ at $N=100$), and its
undetected window is the one-sided $[0,2M]$ of the case $c=1$ in Theorem~\ref{thm:nondeg}(i).

\begin{figure}[!ht]
  \centering
  \includegraphics[width = 0.85\textwidth, height = 5cm]{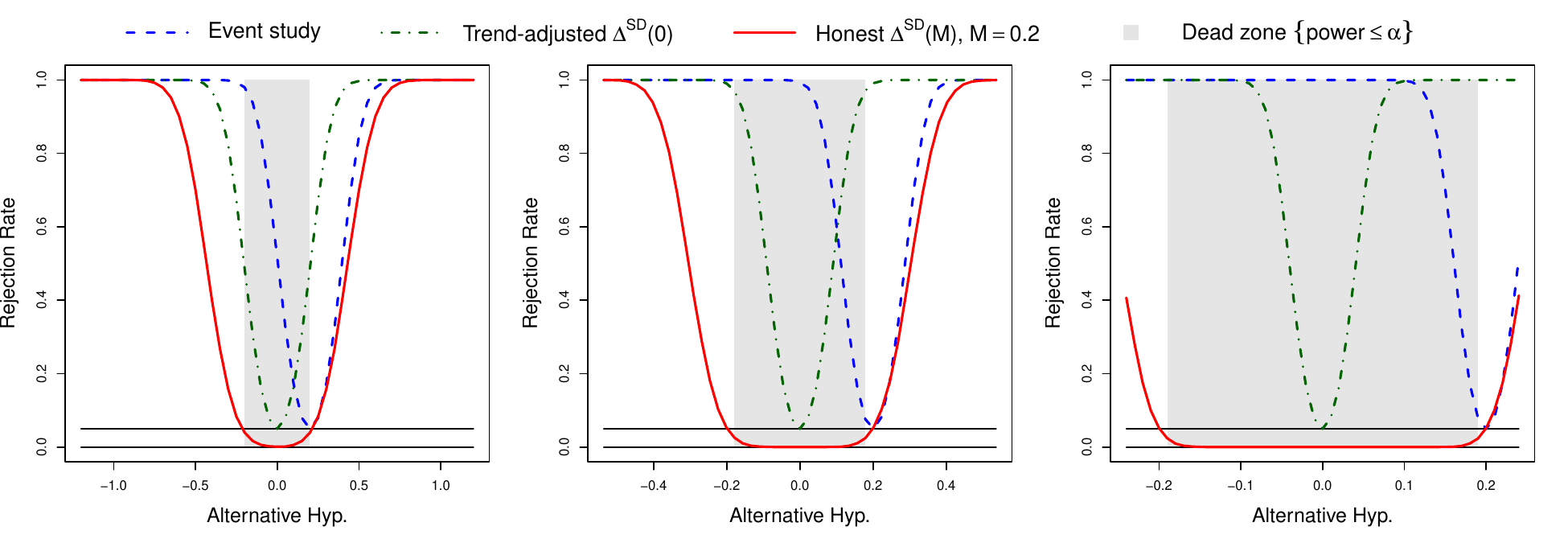}
  \centering
  \includegraphics[width = 0.85\textwidth, height = 5cm]{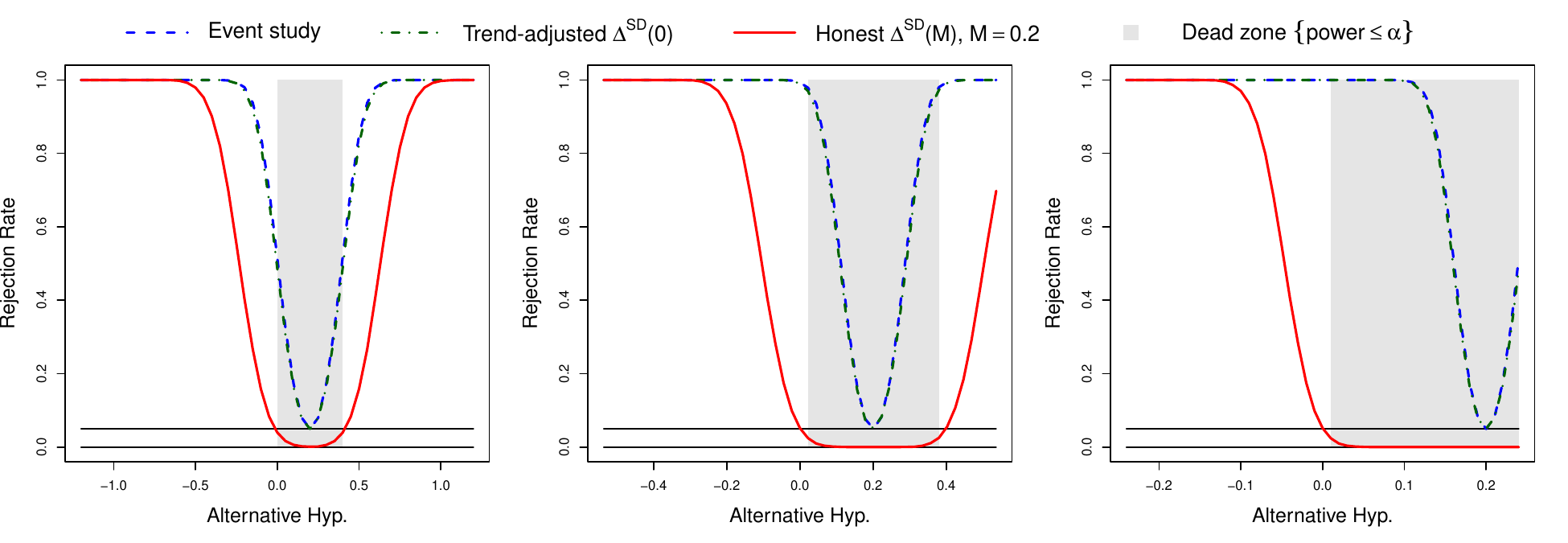}
  \centering
  \includegraphics[width = 0.85\textwidth, height = 5cm]{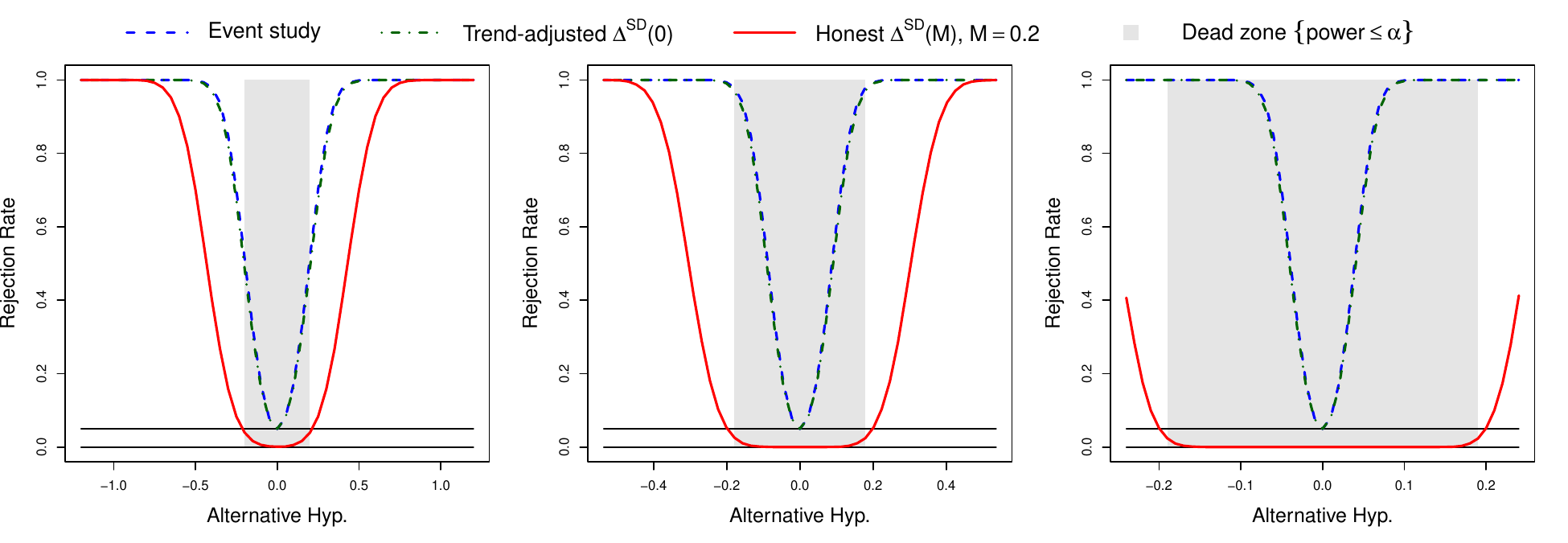}
  \caption{Honest DiD rejection rates for $N^{-1/2}$ local alternatives.
  Columns are $N\in\{100,500,2500\}$, rows the strong pre-trend, kinked violation and no
  pre-trend. Lines show the event-study interval (blue), the trend-adjusted
  $\Delta^{SD}(0)$ interval (green) and the honest $\Delta^{SD}(M)$ interval with
  $M=0.2$ (red). Shading marks the realised dead zone $\{$power $\le\alpha\}$.
  }
  \label{fig:did}
\end{figure}

\subsection{Honest regression discontinuity}
\label{subsec:rd}

The sharp regression-discontinuity design follows the supplement of
\citet{armstrong2018optimal, armstrong2020simple}, with $x_i\sim\mathrm{Uniform}[-1,1]$,
$\varepsilon_i\sim\mathrm{i.i.d.}\,N(0,\sigma^2)$, $\sigma^2=0.1295$, and
\begin{equation*}
  y_i=\tau\,\mathbf{1}\{x_i\ge 0\}+ a\,f(x_i)+\varepsilon_i,
\end{equation*}
where the true jump is $\tau=0$ and $f(x)=C\,\operatorname{sign}(x)\,x^2$, $C=1$, is the
least-favourable odd quadratic, saturating $|f''|=2C$, scaled by $a\in[0,1]$. Local alternatives
are $A_n=\tau\pm\xi\,n^{-2/5}$. Both intervals use the triangular-kernel local-linear estimate
$\hat\tau$ at a common bandwidth. The \emph{conventional} interval
$\hat\tau\pm z_{1-\alpha/2}\,\mathrm{se}(\hat\tau)$ drops the smoothing bias, and the
\emph{honest} interval replaces $z_{1-\alpha/2}$ by the bias-aware critical value
$\mathrm{cv}_\alpha(\bar B/\mathrm{se})$, with $\bar B$ the worst-case bias over the H\"older
class $\{|f''|\le M\}$, $M=2C$.

Figure~\ref{fig:rd} shows the matched-rate case of Theorem~\ref{thm:nondeg}
($\delta=\rho=\epsilon=2/5$), in which conservatism costs only a bounded amount. At the optimal
bandwidth $\bar B/\mathrm{se}\approx0.56$, so the critical value is $2.23$ rather than $1.96$
\citep{armstrong2020simple}, the conventional interval covers only about $0.91$ at $a=1$, and at
$a=0$, where the conventional interval is correct, the honest interval is about $14\%$ wider. The larger under-coverage documented by
\citet{armstrong2018optimal} arises with data-driven bandwidths that mis-estimate curvature.

\begin{figure}[!ht]
  \centering
  \includegraphics[width=0.85\textwidth, height = 5cm]{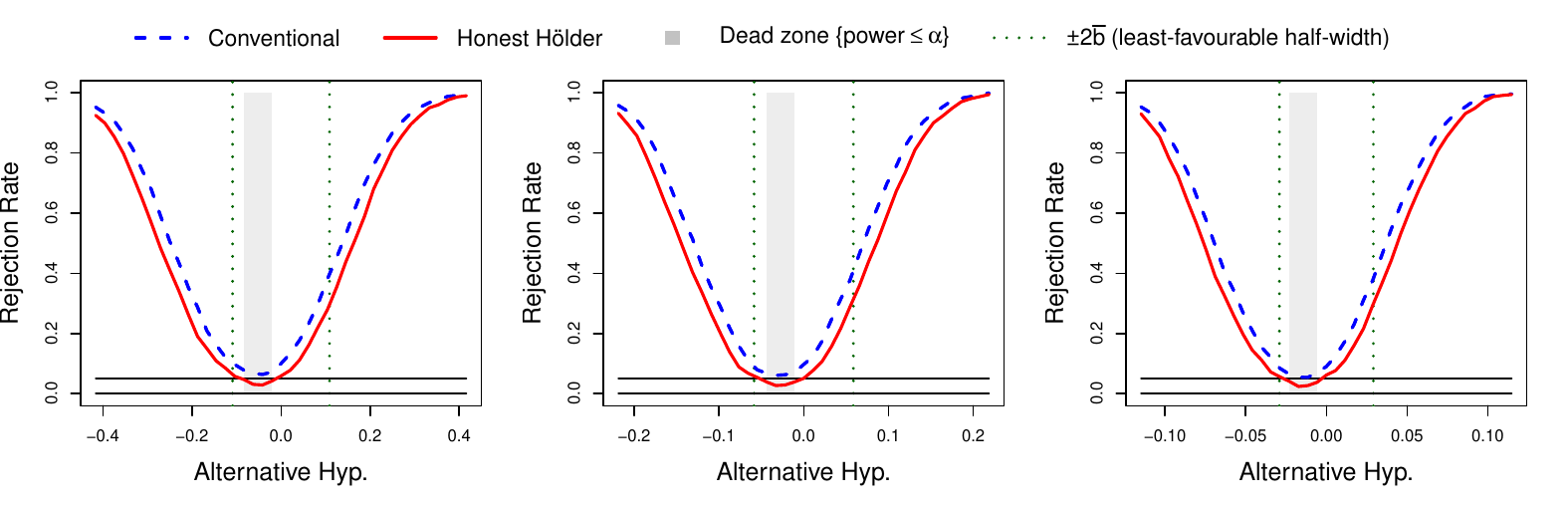}
  \\
  \centering
  \includegraphics[width=0.85\textwidth, height = 5cm]{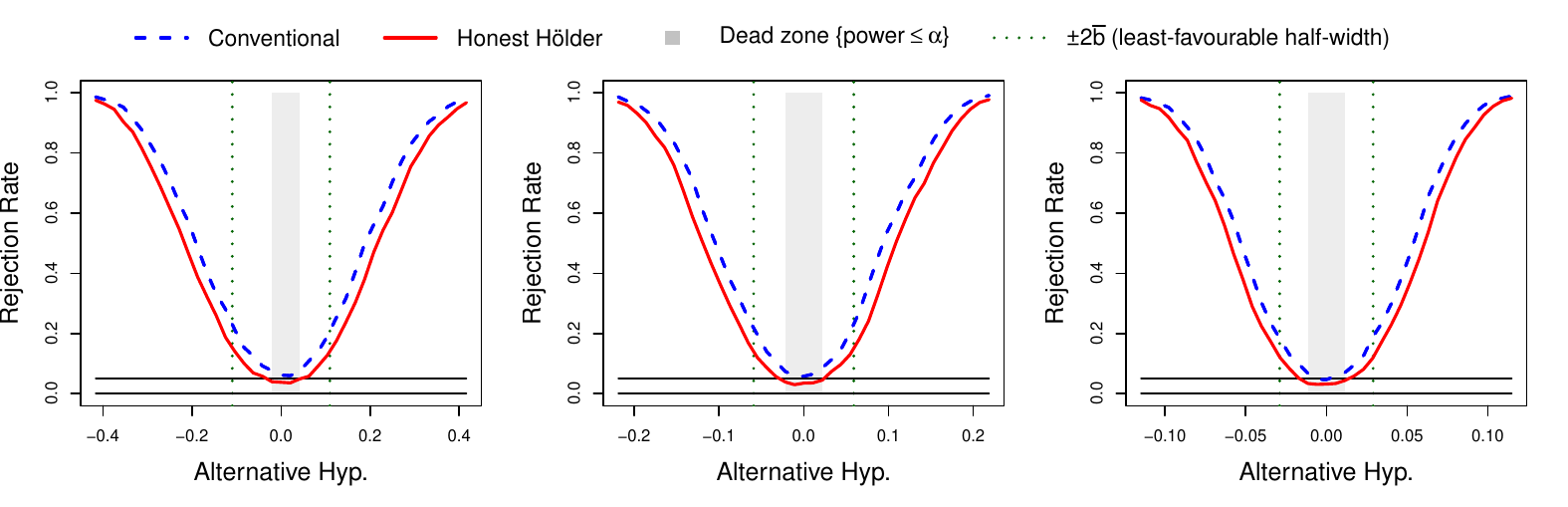}
  \caption{Honest RD rejection rates for $n^{-2/5}$ local alternatives.
  Columns are $n\in\{500, 2500, 12500\}$, rows $a=1$ and $a=0$. Lines show the conventional
  (blue dashed) and honest H\"older (red solid) intervals.
  Shading marks the realised dead zone $\{$power $\le\alpha\}$, and green dotted lines mark
  $\pm2\bar b$, the least-favourable dead-zone half-width for tests built from the estimator
  (Proposition~\ref{prop:sharp}).
  }
  \label{fig:rd}
\end{figure}

\subsection{Estimation of a linear functional}
\label{subsec:donoho}

The stylised problem of \cite{donoho1994statistical} has $x_i\sim\mathrm{Uniform}[0,1]$,
$i=1,\dots,n$, and
\begin{equation*}
  y_i=f(x_i)+\varepsilon_i,\qquad
  f(x)=\theta+a\,C\,\min(x,h),\qquad \varepsilon_i\sim\mathrm{i.i.d.}\,N(0,\sigma^2),
\end{equation*}
with $\sigma=1$ and target the boundary value $f(0)=\theta=0$. The Lipschitz wedge of amplitude
$a\in[0,1]$ rises over $[0,h]$, where $h$ is the honest bandwidth defined below, and is flat
beyond it. Local alternatives are $A_n=\theta\pm\xi/\sqrt n$. The \emph{parametric} interval
$\hat\theta_{ols}\pm z_{1-\alpha/2}\,\mathrm{se}(\hat\theta_{ols})$ uses the OLS intercept,
assuming the approximately linear model of \citet[\S9.1]{donoho1994statistical} holds exactly.
The \emph{honest} interval assumes only the Lipschitz class $\{|f(x)-f(x')|\le C|x-x'|\}$, $C=2$.
It widens the one-sided local average $\hat\theta_{loc}$ over $[0,h]$, with rate-optimal bandwidth
$h=(\sigma^2/(C^2 n))^{1/3}$, by the bias bound $C\,h$, giving
$\hat\theta_{loc}\pm(C\,h+z_{1-\alpha/2}\,\mathrm{se}(\hat\theta_{loc}))$.

At $a=0$ the parametric interval is valid and sharper, and at $a>0$ it undercovers $f(0)$
(Figure~\ref{fig:donoho}). The wedge is untestable, perturbing only about $n^{2/3}$ observations
by about $n^{-1/3}$, so a $y\sim x+x^2$ test at level $0.05$ rejects with probability at most about $0.06$ for every $n$ and the
conservatism cannot be sidestepped by pre-testing. The honest interval shrinks only at
$n^{-1/3}$, so against $n^{-1/2}$ alternatives its rejection probability is asymptotically flat
in $\xi$ (Lemma~\ref{lem:containment}).

\begin{figure}[!ht]
   \centering
  \includegraphics[width = 0.85\textwidth, height = 5cm]{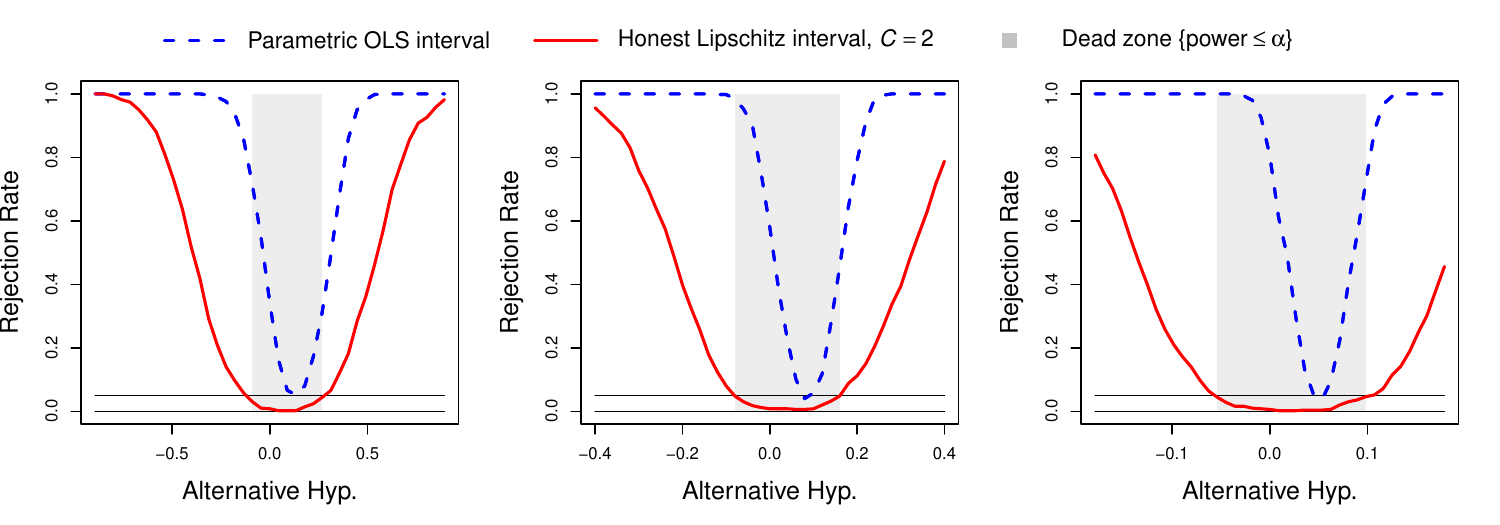}
  \centering
  \includegraphics[width = 0.85\textwidth, height = 5cm]{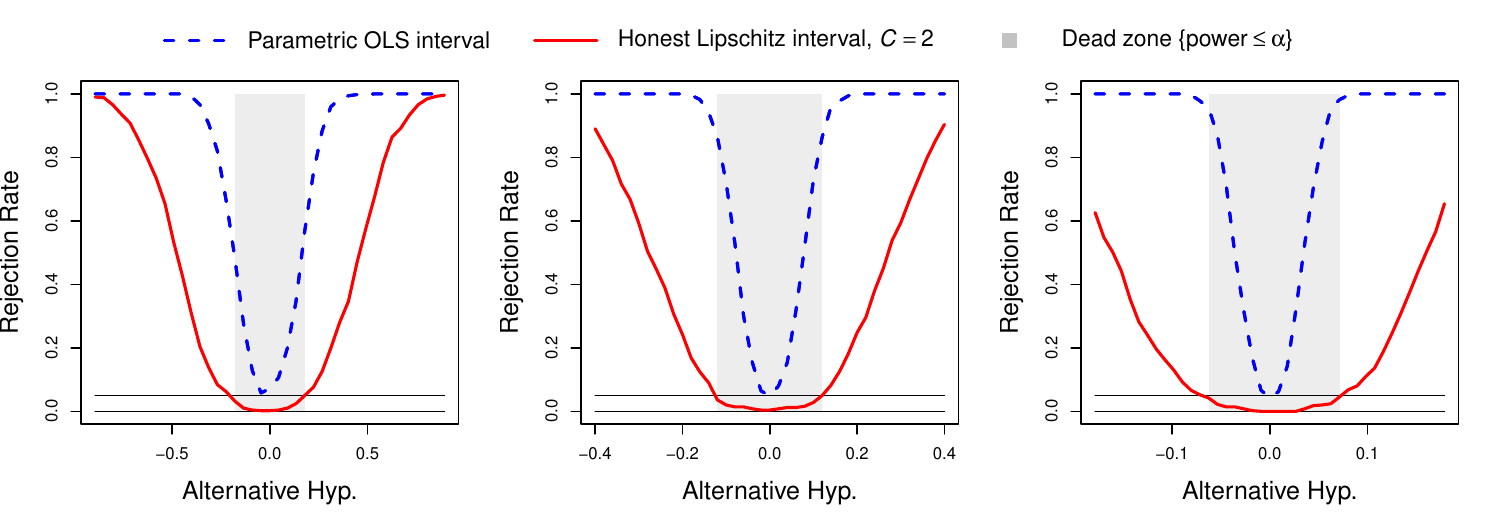}
  \caption{Linear-functional rejection rates for $n^{-1/2}$ local alternatives.
  Columns are $n\in\{500,2500,12500\}$, rows $a=1$ and $a=0$. Lines show the
  parametric OLS interval (blue dashed) and the honest Lipschitz interval with $C=2$ (red
  solid). Shading marks the realised dead zone $\{$power $\le\alpha\}$.
  }
  \label{fig:donoho}
\end{figure}

\subsection{Plausibly exogenous instrumental variables}
\label{subsec:plausexog}

The plausibly-exogenous IV design of \cite{conley2012plausibly} has
\begin{equation*}
  z_i\sim N(0,1),\quad
  x_i=\pi z_i+v_i,\quad
  y_i=\beta x_i+\gamma z_i+\varepsilon_i,\quad
  \varepsilon_i=\rho\, v_i+\sqrt{1-\rho^2}\,\nu_i,
\end{equation*}
with $v_i,\nu_i\sim\mathrm{i.i.d.}\,N(0,1)$, $\pi=1$, $\rho=0.6$ and $\beta=1$. The analyst allows
$|\gamma|\le\bar\gamma=0.30$, and the true violation $\gamma\in\{0.3,0.05,0\}$ gives a strongly
endogenous, a weakly endogenous and an exogenous instrument. Local alternatives are
$A_n=\beta\pm\xi/\sqrt n$. \emph{Standard} IV imposes $\gamma=0$, with
$\hat\beta_{IV}=(z^\top y)/(z^\top x)$ and interval
$\hat\beta_{IV}\pm z_{1-\alpha/2}\,\mathrm{se}(\hat\beta_{IV})$. The \emph{honest} interval is the
union of intervals of \cite{conley2012plausibly},
$\big[\min_{|g|\le\bar\gamma}\mathrm{lb}(g),\ \max_{|g|\le\bar\gamma}\mathrm{ub}(g)\big]$, with
$[\mathrm{lb}(g),\mathrm{ub}(g)]$ the IV interval centred at
$\hat\beta(g)=z^\top(y-g z)/(z^\top x)$.

The union widens the interval by an amount of order
$\bar\gamma\,(z^\top z)/(z^\top x)\approx\bar\gamma/\pi$, constant in $n$, so the dead zone does
not shrink as $n$ grows (Figure~\ref{fig:plausexog}).

\begin{figure}[!ht]
  \centering
  \includegraphics[width = 0.85\textwidth, height = 4.8cm]{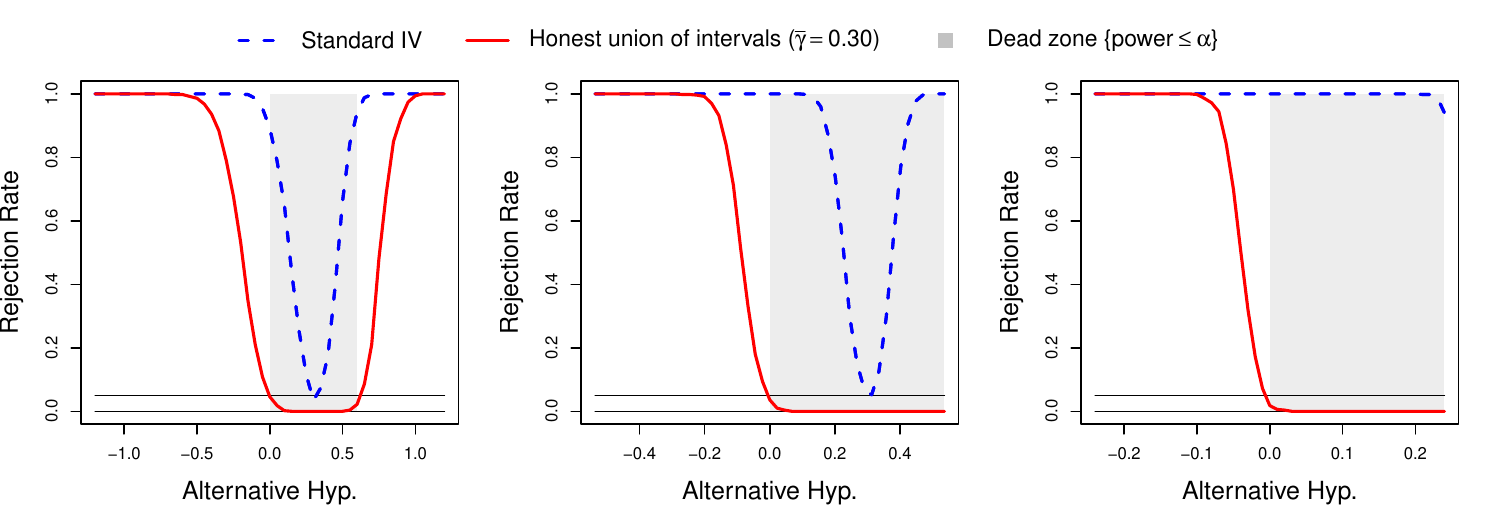}
  \centering
  \includegraphics[width = 0.85\textwidth, height = 4.8cm]{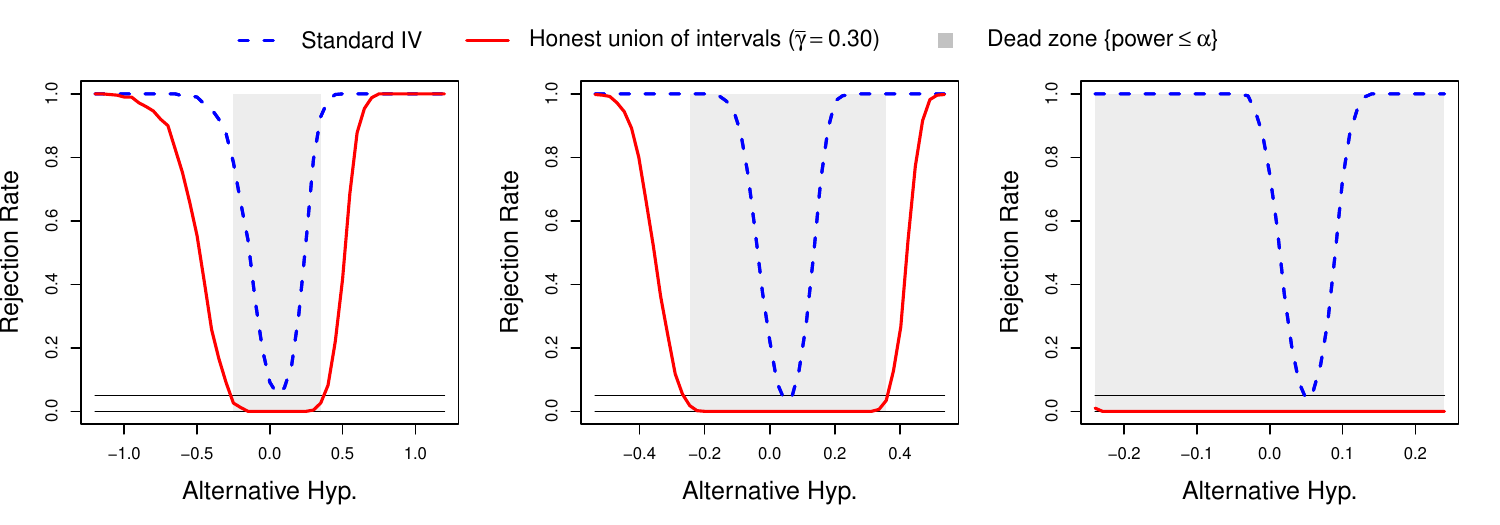}
  \centering
  \includegraphics[width = 0.85\textwidth, height = 4.8cm]{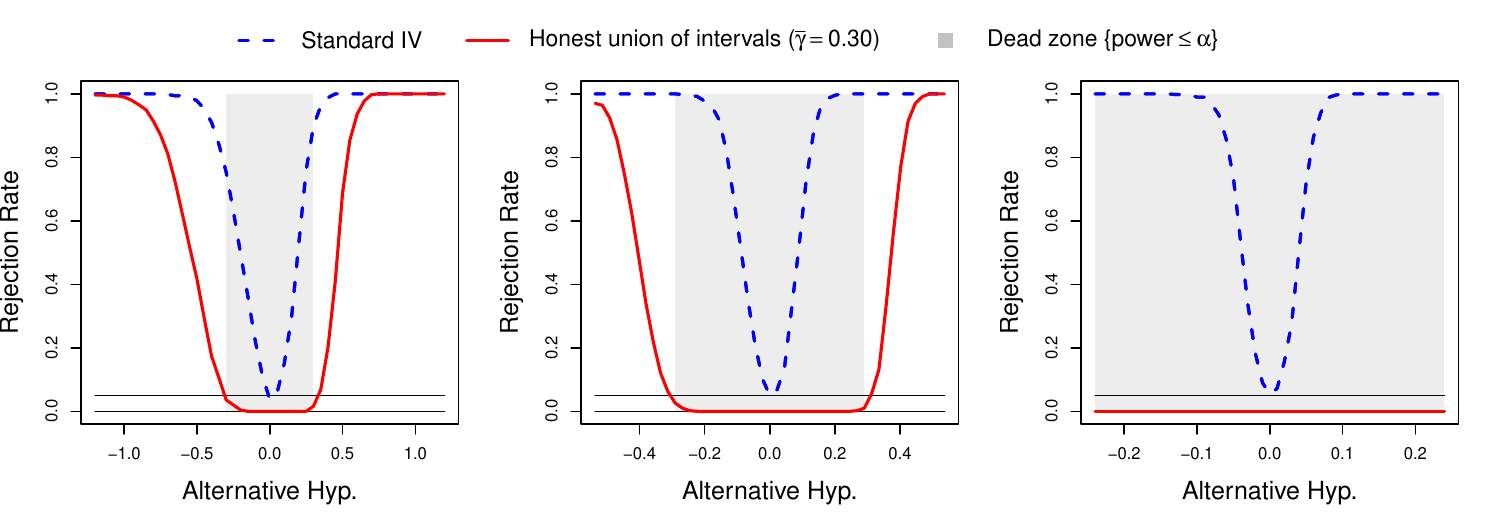}
  \caption{Plausibly-exogenous IV rejection rates for $n^{-1/2}$ local alternatives.
  Columns are $n\in\{100,500,2500\}$, rows the strongly endogenous ($\gamma=0.3$), weakly
  endogenous ($\gamma=0.05$) and exogenous ($\gamma=0$) instrument. Lines show standard IV (blue
  dashed) and the honest union of intervals with $\bar\gamma=0.30$ (red solid). Shading marks
  the realised dead zone $\{$power $\le\alpha\}$.
  }
  \label{fig:plausexog}
\end{figure}

\section{Conclusion and recommendations}\label{sect:discussion}

This paper establishes a local-alternative power trade-off intrinsic to conservative
inference. At the usual parametric rate, an honest interval may have no local power, and
the loss is shared by every honest procedure. The paper is
closed with practical recommendations. 

First, $\bar B=\sup_\theta|\hat B|$ is chosen by the analyst, so the comparison merely
re-expresses how conservative a bound is. 
The power cost is not a hidden risk to
be detected but a deterministic consequence of a stated belief about unobservable
biases, i.e. a monotone function of that preference. What this analysis recommends is
\emph{disclosure}. Report the dead-zone half-width, $2\bar B$, alongside any bias-aware interval, since it states exactly which alternatives
the chosen bound forgoes the ability to detect. This is information the interval's endpoints
alone do not convey.

Second, when the bound dominates standard errors, typical in the parametric rate case, and the analyst
nonetheless needs power, the only escape is to relax strict honesty.
For instance, to defend a
\emph{tighter} bound on substantive grounds, or to report an undersized test that retains
local power \citep[cf.][]{low1997nonparametric,rambachan2023more}. Which is preferable is
a genuine modelling choice that the size--power frontier here makes explicit rather than
resolves.
Conservatism buys uniform coverage at a power cost that is, at the parametric
rate, total. Quantifying when the bound can be credibly tightened, and characterising the
optimal undersized test are natural next steps for research.

\setlength{\bibsep}{2pt} 
\bibliographystyle{chicago3}
\bibliography{refs}

\end{document}